\documentclass{article}

\usepackage[margin=1in]{geometry}

\usepackage{amssymb,amsmath,amsthm,amsfonts}
\usepackage{bm}
\usepackage{textgreek}
\usepackage{mathtools}
\usepackage{enumitem}
\usepackage[numbers,comma,sort&compress]{natbib}
\usepackage{authblk}
\usepackage{graphicx}
\usepackage[font=small]{caption}
\usepackage{subcaption}
\usepackage{tablefootnote} 
\usepackage{float}
\usepackage[ruled,vlined,linesnumbered]{algorithm2e}

\usepackage[colorlinks]{hyperref}

\usepackage{braket}           
\usepackage{overpic}
\usepackage{multirow}
\usepackage{xcolor}
\let\Horig\H

\newtheorem{theorem}{Theorem}
\newtheorem{definition}{Definition}

\newtheorem{proposition}{Proposition}

\newtheorem{corollary}{Corollary}

\newcommand{\eq}[1]{(\ref{eq:#1})}

\newcommand{\thm}[1]{\hyperref[thm:#1]{Theorem~\ref*{thm:#1}}}
\newcommand{\cor}[1]{\hyperref[cor:#1]{Corollary~\ref*{cor:#1}}}
\newcommand{\defn}[1]{\hyperref[defn:#1]{Definition~\ref*{defn:#1}}}
\newcommand{\lem}[1]{\hyperref[lem:#1]{Lemma~\ref*{lem:#1}}}
\newcommand{\prop}[1]{\hyperref[prop:#1]{Proposition~\ref*{prop:#1}}}
\newcommand{\assum}[1]{\hyperref[assum:#1]{Assumption~\ref*{assum:#1}}}
\newcommand{\fig}[1]{\hyperref[fig:#1]{Figure~\ref*{fig:#1}}}
\newcommand{\tab}[1]{\hyperref[tab:#1]{Table~\ref*{tab:#1}}}
\newcommand{\algo}[1]{\hyperref[algo:#1]{Algorithm~\ref*{algo:#1}}}
\renewcommand{\sec}[1]{\hyperref[sec:#1]{Section~\ref*{sec:#1}}}
\newcommand{\append}[1]{\hyperref[append:#1]{Appendix~\ref*{append:#1}}}
\newcommand{\fac}[1]{\hyperref[fac:#1]{Fact~\ref*{fac:#1}}}
\newcommand{\lin}[1]{\hyperref[lin:#1]{Line~\ref*{lin:#1}}}

\def\>{\rangle}
\def\<{\langle}

\renewcommand{\H}{\mathcal{H}}

\title{Quantum Approximate Optimization Algorithms for Maximum Cut on Low-Girth Graphs}
\author[1,2]{Tongyang Li\thanks{tongyangli@pku.edu.cn}}
\author[1,2]{Yuexin Su\thanks{yuexinsu@stu.pku.edu.cn}}
\author[3]{Ziyi Yang\thanks{2100010833@stu.pku.edu.cn}}
\author[4]{Shengyu Zhang\thanks{shengyzhang@tencent.com}}

\affil[1]{School of Computer Science, Peking University}
\affil[2]{Center on Frontiers of Computing Studies, Peking University}
\affil[3]{School of Mathematical Science, Peking University}
\affil[4]{Tencent Quantum Laboratory}

\date{}

\begin{document}

\maketitle

\begin{abstract}
Maximum cut (MaxCut) on graphs is a classic NP-hard problem. In quantum computing, Farhi, Gutmann, and Goldstone proposed the Quantum Approximate Optimization Algorithm (QAOA) for solving the MaxCut problem. Its guarantee on cut fraction (the fraction of edges in the output cut over all edges) was mainly studied for high-girth  graphs, i.e., graphs with only long cycles. On the other hand, low-girth graphs are ubiquitous in theoretical computer science, including expander graphs being outstanding examples with wide applications in theory and beyond. 
In this paper, we apply QAOA to MaxCut on a set of expander graphs proposed by Mohanty and O'Donnell known as additive product graphs. Additionally, we apply multi-angle QAOA (ma-QAOA) to better utilize the graph structure of additive product graphs in ansatz design.
In theory, we derive an iterative formula to calculate the expected cut fraction of such graphs. This formula also extends to the quantum MaxCut problem.
On the other hand, we conduct numerical experiments to compare between best-known classical local algorithms and QAOA with constant depth. 
Our results demonstrate that QAOA outperforms the best-known classical algorithms by 0.3\% to 5.2\% on several additive product graphs, while ma-QAOA further enhances this advantage by an additional 0.6\% to 2.5\%. In particular, we observe cases that ma-QAOA exhibits superiority over best-known classical algorithms but QAOA does not.
Furthermore, we extend our experiments to planar graphs such as tiling grid graphs, where QAOA also demonstrates an advantage. 
\end{abstract}

\section{Introduction}

MaxCut is one of the most fundamental problems in graph theory. The input of the problem is a simple (unweighted, undirected) graph, and the goal is to partition the vertices into two sets such that the number of edges between the two parts is maximized. MaxCut is one of the 21 NP-complete problems shown by Karp in 1972~\cite{karp2010reducibility}. Nevertheless, Goemans and Williamson~\cite{goemans1995improved} gave an approximation algorithm for MaxCut with an expected solution being at least 0.878 times the number of cut edges of the optimal solution. Meanwhile, solving MaxCut problem beyond this approximation ratio (the number of cut edges divided by the number of maximal cut edges) 0.878 will imply that the Unique Games Conjecture does not hold~\cite{khot2007optimal}. Therefore, this approximation ratio value of 0.878 can be seen as an essential barrier for the MaxCut problem, and a theoretical point of perspective, an approximation algorithm for MaxCut is better if its approximation ratio is closer to 0.878.

More recently, quantum algorithms have become promising candidates in solving classical combinatorial problems, with the quantum approximate optimization algorithm (QAOA) introduced in \cite{farhi2014quantum}. QAOA is a quantum algorithm that produces approximate solutions to combinatorial optimizations problems. 
For the MaxCut problem, QAOA applies unitary evolutions of the sum-of-Pauli $X$ operator and the sum-of-Pauli $ZZ$ operator that encode the graph edges (see a formal definition in \sec{prelim-QAOA}) alternatively for $p$ times.
The quality of solution improves as $p$ increases, and Ref.~\cite{farhi2014quantum} proved that the optimal solution can be obtained when $p$ approaches infinity. On the other hand, it is observed that QAOA with constant $p$ can already guarantee good approximation for MaxCut. For 3-regular graphs, Ref.~\cite{farhi2014quantum} proved that for $p=1$, QAOA can get approximation ratio of at least 0.6924.\footnote{Note that there exists a 1-local quantum algorithm based on annealing that achieves an approximation ratio over 0.7020 on 3-regular graphs \cite{braida2024tight}. The $p$-local analysis of quantum annealing involves only analyzing it on subgraphs composed of the vertices within distance $p$ from each edge. In this paper, we specifically focus on QAOA with constant $p$.}
For $p=2$ and $p=3$, Ref.~\cite{wurtz2021maxcut} proved that QAOA can achieve approximation approximation ratio 0.7559 and 0.7924 for the MaxCut problem, respectively.

For more general graphs, previous literature on QAOA with approximation ratio guarantee for MaxCut focuses on graphs with high girth. In graph theory, the \emph{girth} of an undirected graph is the length of a shortest cycle contained in the graph. Hastings~\cite{hastings2019classical} compared single-step classical algorithms and $\text{QAOA}$ with $p=1$ on $D$-regular triangle-free graphs (i.e., graphs with girth $>3$), and found that single-step classical algorithms outperform $\text{QAOA}$ with $p=1$ for $3 \leq D \leq 1000$ except 4 choices of degree $D$. Furthermore, Marwaha~\cite{marwaha2021local} studied two-step classical algorithms on $D$-regular graphs with girth $>5$, and found that the optimal two-step threshold classical algorithm outperforms $\text{QAOA}$ with $p=2$ for all $5 < D \leq 500$. 
On the other hand, Basso et al.~\cite{basso2022quantum} analyzed the cut fraction (the number of cut edges divided by the number of total edges) of QAOA for MaxCut on $D$-regular graphs with girth $>2p+1$. They gave an iterative formula to evaluate the performance of QAOA for any fixed $p$ and $D$. With $p \geq 11$, QAOA outputs a cut with larger cut fraction than best-known classical algorithms~\cite{barak2022classical}.

Previous studies focus on high-girth $d$-regular graphs because QAOA is a local algorithm and for high-girth graphs, all local neighborhoods are two complete $d$-ary trees with roots glued together, which facilitates analysis. 
Random high-girth regular graphs are sparse \emph{expander graphs} with high probability, which are sparse graphs with strong connectivity, and they also have wide applications in theoretical computer science including algorithm design, error correcting codes, extractors, pseudorandom generators, sorting networks, etc~\cite{lubotzky2012expander,hoory2006expander}. While the performance of QAOA on high-girth graphs has been extensively studied, the approximation ratio guarantee of QAOA on expander graphs with low girth remains widely open. We extend previous studies to a set of low-girth expander graphs proposed by Mohanty and O’Donnell \cite{mohanty2020x}
known as  \emph{additive product graphs} (formally defined in \sec{x-ramanujan}). The additive product graphs are Ramanujan graphs, i.e., 
there are infinitely many graphs $G$ that are covered by them and their spectrum satisfies the second largest eigenvalue is no greater than the spectral radius of additive product graphs.
Such graphs are low-girth graphs with small cycles, of general interest to the theoretical computer science, but previous work on the cut fraction of QAOA requires high girth in analysis. Understanding of how QAOA performs on low-girth graphs with additive product graphs as examples is solicited, and such understanding may offer an insight on how QAOA performs on a variety of expander graphs.

\paragraph{Main results.} In this paper, we systematically investigate the cut fraction of QAOA for low-girth graphs. 
Moreover, to better utilize the graph structure, we discuss ma-QAOA (see definition and specific settings in \sec{ma-qaoa}) to enhance the advantage of QAOA. 

In theory, we apply QAOA to MaxCut on additive product graphs. We give an iterative formula to evaluate the expected cut fraction of additive product graphs for any fixed $p$. 
The derivation of the iterative formula for additive product graphs is based on analyzing the underlying construction of the graph and identifying the relevant subgraphs. 
\begin{theorem}[Main Theorem]
Suppose $X$ is an additive product graph defined in \defn{additive-graph}. Then for any $p$ and any parameters $(\boldsymbol{\gamma},\boldsymbol{\beta})\in [0,2\pi]^{2p}$, the expected cut fraction of the additive product graph satisfies
\begin{equation}
    \frac{\bra{\boldsymbol{\gamma},\boldsymbol{\beta}} C_{MC} \ket{\boldsymbol{\gamma},\boldsymbol{\beta}}}{|E|}
    = \frac12 + \frac{1}{2(|E({\underline{A}_1})|+\cdots+|E({\underline{A}_c})|)} \sum_{\substack{C \in [c], \\ (a,b) \in \underline{A}_C}} 
    \mathbb{E}[\underline{A}_C(a,b)],
\end{equation}
where the expectations $\mathbb{E}[\underline{A}_C(a,b)]$ are defined on different subgraphs and follow the recursive formulas in \eq{EA} and \eq{G}.
\end{theorem}
To the best of our knowledge, this is the first QAOA result that considers low-girth graphs with general $p$ with theoretical guarantee. Previous research has primarily focused on studying regular graphs of small values of $p$ \cite{wurtz2021maxcut} and high-girth regular graphs of general $p$ \cite{basso2022quantum}.

Technically, we extend the analysis in \cite{basso2022quantum} from high-girth regular graphs to the more general setting of graphs that do not necessarily have high-girth. Specifically, we identify a class of spectral expander graphs, namely additive product graphs, in which the structure of subgraphs can be analyzed systematically. By leveraging this subgraph structure, we are then able to analyze the cycles in low-girth graphs iteratively and obtain the expected cut fraction of such low-girth graphs. Furthermore, we extend the analytical technique to the quantum MaxCut problem on additive product graphs, extending the existing literature on quantum MaxCut for high-girth regular graphs \cite{kannan2024quantum}.

In experiments, we explore classical $k$-local algorithms for MaxCut on additive product graphs to provide convincing benchmarks against QAOA. We explore the performance of the best-known classical local algorithms previously utilized in evaluating QAOA's performance on high-girth regular graphs. Specifically, we evaluate the classical algorithms proposed by Barak and Marwaha \cite{barak2022classical} as well as the threshold algorithm demonstrated in \cite{hirvonen2014large,hastings2019classical,marwaha2021local}. Our findings indicate that on low-girth graphs, Barak and Marwaha's algorithm yields a superior expected cut fraction at lower values of $k$, whereas the threshold algorithm exhibits stronger performance at higher $k$ values.

We conduct extensive numerical evaluations for QAOA and the best-known classical local algorithms for the cut fraction of additive product graphs. Because low-girth graphs take significantly more computing resources, we only conduct numerical experiments for $p=1,2,3$. 
We found that for \fig{a} and \fig{c} tested under lower values of $p$, the QAOA algorithm outperformed the best-known classical local algorithms by 0.3-5.2 percentage and ma-QAOA enhances this advantage by 0.6-2.5 percentage. Moreover, although the classical local algorithms exhibit minor advantage against QAOA in specific graph instances \fig{b}, ma-QAOA surpasses the best-known classical local algorithms by 3.0 percentage at $p=2$. We demonstrate our results in \fig{experi}.
\begin{figure}[H]
    \centering
    \includegraphics[scale=0.5]{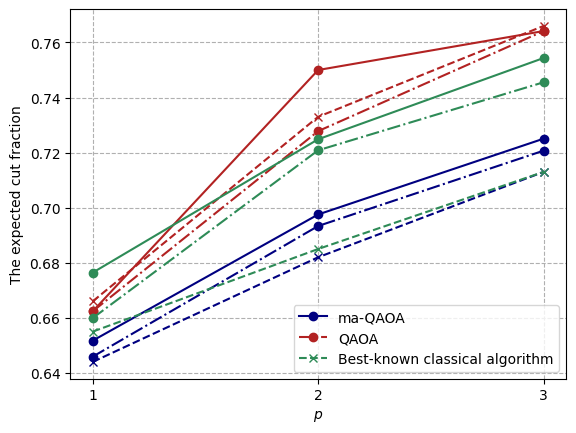}
    \caption{The expected cut fractions of QAOA and the best-known classical local algorithms are illustrated for the graphs shown in \fig{graph-set1}. The results in \fig{a}, \fig{b}, and \fig{c} correspond to the colors blue, red, and green, respectively. The solid lines, dash-dotted lines, and dashed lines represent the performance of ma-QAOA, QAOA, and the best-known classical local algorithms discussed in \sec{experiments}.}
    \label{fig:experi}
\end{figure}

Furthermore, we extend our experiments to planar graphs composed of polygons. The main difference between this type of planar graph and the additive product graph is the manner in which they extend. Additive product graphs extend in a more tree-like hierarchical structure, whereas the planar graphs we study exhibit a more mesh-like, interconnected structure. Thus far, we have investigated the expected cut fraction of the planar graph composed of pentagons and hexagons shown in \fig{5-6cycle}, as well as the planar graph composed of triangles, quadrilaterals, and hexagons presented in \fig{3-4-6cycle}. Given the computational complexity constraints, we have run the numerical experiments for $p={1,2}$. 
The results demonstrate that the QAOA still outperforms the classical algorithms by 0.3-2.2 percentage, and ma-QAOA enhances this advantage by 0.1-0.9 percentage.

\section{Preliminaries}

\subsection{The quantum approximate optimization algorithm and MaxCut}\label{sec:prelim-QAOA}
The quantum approximate optimization algorithm was introduced in \cite{farhi2014quantum}, which is a variational quantum algorithm that requires $2p$ parameters: $(\gamma_1, \gamma_2, \ldots, \gamma_p)$, $(\beta_1, \beta_2, \ldots, \beta_p)$. The input is an $n$-qubit string $z$ and the goal is to find an approximate ground state of cost function operator $C$, where $C\ket{z} = C(z) \ket{z}$. It does so by preparing the state $\ket{\boldsymbol{\gamma},\boldsymbol{\beta}}$ on the ground state $\ket{s}$ of mixing operator $B$. 
QAOA prepares the state
\begin{align}\label{eq:QAOA}
    \ket{\boldsymbol{\gamma},\boldsymbol{\beta}} = 
    U(B,\beta_p)U(C,\gamma_p)\cdots U(B,\beta_1)U(C,\gamma_1)
    \ket{s}.
\end{align}
where $U(B,\beta)=e^{-i\beta B}$ and $U(C,\beta)=e^{-i\gamma C}$.

The expectation of the cost function $C$ is $\bra{\boldsymbol{\gamma},\boldsymbol{\beta}} C \ket{\boldsymbol{\gamma},\boldsymbol{\beta}}$. When $p$ approaches infinity, Eq.~\eq{QAOA} can be seen as Trotterization of the adiabatic theorem, and thus it can reach the minimum of the cost function operator $C$ \cite{farhi2014quantum}. In practice, for a fixed value of $p$, we can measure $\ket{\boldsymbol{\gamma},\boldsymbol{\beta}}$ in computational basis and optimize $\boldsymbol{\gamma}$ and $\boldsymbol{\beta}$.

Given a graph $G=(V,E)$, the cost function of MaxCut is to evaluate how many edges are cut due to the partition of vertices. 
If the qubits of vertices $u$ and $v$ in edge $(u,v)$ are different, then $Z_uZ_v=-1$ and count one edge to the cost function. If the qubits are the same, $\frac12(-Z_uZ_v+1)=0$. Thus the operator can be written as:
\begin{equation}
\label{eq:maxcut-operator}
    C_{MC} = \sum_{(u,v)\in E}\frac12(-Z_uZ_v+1),
\end{equation}
where $Z_u$ is Pauli $Z$ operator on qubit $u$.

Since the constant $1/2$ in $C_{MC}$ only introduces a global phase that does not influence measurements, we can instead use the scaled cost function operator
\begin{equation}
    C = -\sum_{(u,v)\in E} Z_uZ_v.
\end{equation}

And the mixing operator $B$ equals 
\begin{equation}
    B = \sum_{v \in V} X_v
\end{equation}

Note that QAOA is a local algorithm where the expectation $\bra{\boldsymbol{\gamma},\boldsymbol{\beta}} Z_uZ_v \ket{\boldsymbol{\gamma},\boldsymbol{\beta}}$ on each edge only depends on the edges whose distance from edge $(u,v)$ are no more than $p$ and qubits on them. Therefore, edges in the same {$p$-neighborhood} subgraph 
have the same expectation values. {That is, $\bra{\boldsymbol{\gamma},\boldsymbol{\beta}} Z_uZ_v \ket{\boldsymbol{\gamma},\boldsymbol{\beta}} = \bra{\boldsymbol{\gamma},\boldsymbol{\beta}} Z_{u'}Z_{v'} \ket{\boldsymbol{\gamma},\boldsymbol{\beta}}$ if the edges $(u,v)$ and $(u',v')$ have the same neighborhood subgraph $g$, in which case we denote the value by $\bra{\boldsymbol{\gamma},\boldsymbol{\beta}} ZZ(g)\ket{\boldsymbol{\gamma},\boldsymbol{\beta}}$.} We can categorize edges in $E$ according to different subgraphs $g$, and the expectation becomes
\begin{equation}
{\bra{\boldsymbol{\gamma},\boldsymbol{\beta}} C_{MC} \ket{\boldsymbol{\gamma},\boldsymbol{\beta}} 
= \sum_{g} \frac{w_g}{2} \left(1-\bra{\boldsymbol{\gamma},\boldsymbol{\beta}} ZZ(g)\ket{\boldsymbol{\gamma},\boldsymbol{\beta}} \right),
}
\end{equation}
where {the summation is over all possible $p$-neighborhood subgraphs and $w_g$ is the number of edges $(u,v)$ whose $p$-neighborhood subgraph is $g$.}
The cut fraction is then
\begin{equation}\label{eq:expectation}
{
\frac{\bra{\boldsymbol{\gamma},\boldsymbol{\beta}} C_{MC} \ket{\boldsymbol{\gamma},\boldsymbol{\beta}}}{|E|}
= \frac12 - \sum_{g} \frac{f_g}{2}
\bra{\boldsymbol{\gamma},\boldsymbol{\beta}} ZZ(g)\ket{\boldsymbol{\gamma},\boldsymbol{\beta}},
}
\end{equation}
where $f_g$ is the proportion of edges with {$p$-neighborhood subgraph} $g$.

\subsection{\texorpdfstring{$X$}{X}-Ramanujan graphs}\label{sec:x-ramanujan}
We will also provide a brief introduction to the $X$-Ramanujan graphs in \cite{mohanty2020x}, the class of expander graphs that is studied by QAOA. First, we introduce the notion of spectral radius:
\begin{definition}\label{defn:spectral-radius}
For a graph $X$, the spectral radius of X is
\begin{equation}
    \rho(X):=\lim_{t\to\infty} \sup \left\{ (c_{uv}^{(t)})^{1/t} \right\},
\end{equation}
where $c_{uv}^{(t)}$ denotes the number of walks of length $t$ in $X$ from vertex $u\in V$ to vertex $v\in V$.
\end{definition}

Then, we introduce the definition of additive product graph. The additive product graph are generated by a set of graphs $A_1,\ldots,A_{c}$, called plain atoms, that share a common vertex set. By choosing a fixed starting vertex $v_1$, the vertices and edges of additive product graph are defined as follows, {with one example illustrated in Figure \ref{fig:additive-graph}.}
\begin{definition}[{\cite[Definition 3.4]{mohanty2020x}}]\label{defn:additive-graph}
Let $A_1,\ldots,A_{c}$ be plain atom {graphs} on {a} common vertex set $[n]$. Define {a} sum graph {$G = A_1+\cdots +A_c = ([n],E)$ with the edge set $E$ being} 
the union of the edge sets of $A_1,\ldots,A_{c}$. Assume that $G$ is connected. Letting $\underline{A}_{j}$ denote $A_j$ with isolated vertices removed, we also assume that each $\underline{A}_{j}$ is nonempty and connected. We now define the (typically infinite) additive product graph $A_1\oplus \cdots \oplus A_c:=(V,E)$ where $V$ and $E$ are constructed as follows.

Let $v_1$ be a fixed vertex in $[n];$ let $V$ be the set of strings of the form $v_1C_1v_2C_2\cdots v_kC_kv_{k+1}$ for $k\geqslant0$ such that:
\begin{enumerate}[label=\arabic*.]
    \item each $v_i$ is in $[n]$ and each $C_i$ is in $[c]$,
    \item $C_i \neq C_{i+1}$ for all $i<k$,
    \item $v_i$ and $v_{i+1}$ are both in $\underline{A}_{C_i}$ for all $i\leq k$,
\end{enumerate}
 and, let $E$ be the set of edges on vertex set $V$ such that for each string $s\in V$
 \begin{enumerate}[label=\arabic*.]
     \item  we let $\{sCu, sCv\} $ be in $E$ if $\{u, v\} $ is an edge in $\underline{A}_{C}$, 
     \item we let $\{sCu,sCuC^{\prime}v\}$ be in $E$ if $\{u,v\}$ is an edge in $\underline{A}_{C^{\prime}}$, and
     \item we let $\{v_1,v_1Cv\}$ be in $E$ if $\{v_1,v\}$ is an edge in $\underline{A}_C.$
 \end{enumerate}
\end{definition}

Note that different choices of $v_1$ generate isomorphic additive product graphs. In \fig{additive-graph}, one example generated by \defn{additive-graph} is depicted, where the generators $A_1,\ldots,A_{c}$ bring corresponding local structures within the additive product graph. Note that an additive product graph may not be regular. A simple example is when $G$ is a complete bipartite graph $K_{\Delta_1,\Delta_2}$, and $A_1,\ldots,A_c$ represent the associated single-edge graphs on $G$’s vertex set, {where $c = \Delta_1 \Delta_2$}. It can be observed that $A_1\oplus A_2 \cdots\oplus A_c$ construct an infinite ($\Delta_1,\Delta_2$)-biregular tree, {where any node in an odd level has degree $\Delta_1$ and any node in an even level has degree $\Delta_2$. See \fig{3-4-biregular-tree} for an illustration of the case of $(\Delta_1,\Delta_2) = (3,4)$.}
\begin{figure}[H]
    \centering
    \includegraphics[scale=0.5,page={3}]{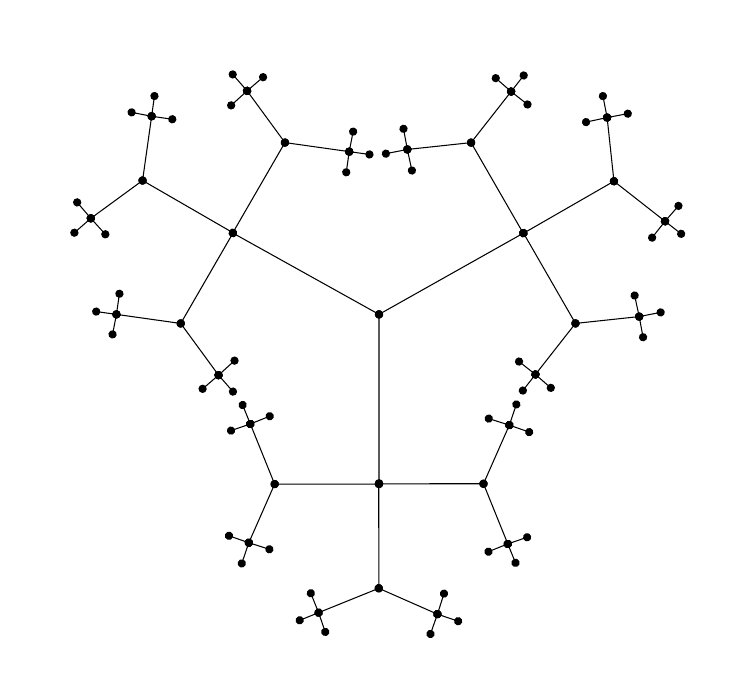}
    \caption{The infinite ($3,4$)-biregular tree constructed by the single-edge graphs on $K_{3,4}$.}
    \label{fig:3-4-biregular-tree}
\end{figure}

{
Additive product graphs can be easily constructed and proved to exhibit certain desirable spectral properties, by which they offer a family of explicit constructions of expander graphs. We refer to \cite{mohanty2020x} for more details on this. }

\begin{figure}[H]
    \centering
    \includegraphics[scale=0.60,page={1}]{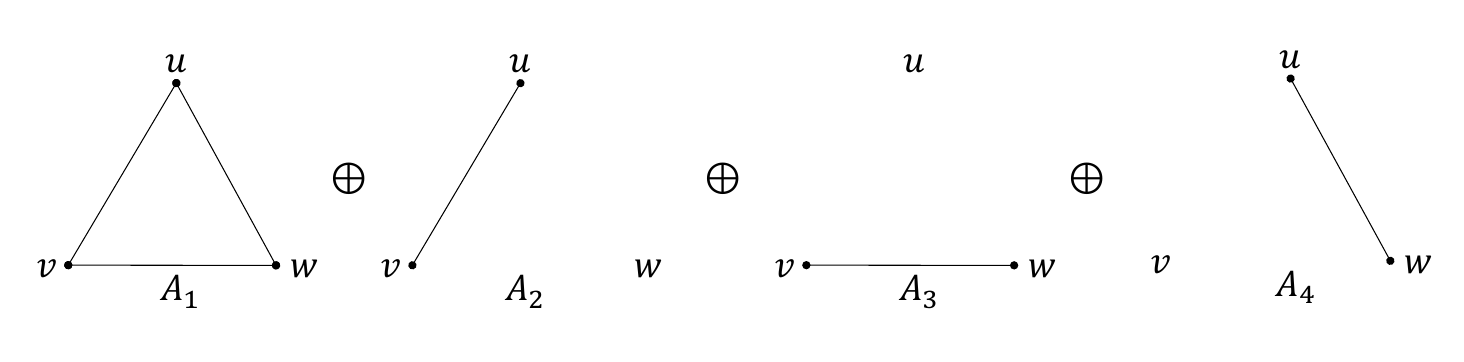}
    \includegraphics[scale=0.60,page={2}]{additive_product_graph.pdf}
    \caption{One example of the additive product graph. A portion of vertices on the graph are labeled with the corresponding string. }
    \label{fig:additive-graph}
\end{figure}

\subsection{The multi-angle quantum approximate optimization algorithm}
\label{sec:ma-qaoa}
{In the original QAOA paper \cite{farhi2014quantum}, the authors let each layer of quantum circuit share the same parameter. Later a multi-angle version of ma-QAOA was proposed \cite{herrman2022multi} in which each node $v$ or edge $(u,v)$ has an individual parameter in each layer $j$, namely $U(\beta,\gamma) = \prod_{j=p}^1 (e^{i\sum_{v\in V} \beta_{jv} X_v} e^{i\sum_{(u,v)\in E} \gamma_{j,(u,v)} Z_u Z_v})$. By introducing more parameters, the model becomes more expressive and can in principle improve the performance of the algorithm. However, more parameters also imply a much larger search space which may make training harder. Indeed, from a dynamical Lie algebra point of view, more parameter sharing shrinks the generated Lie algebra, while smaller dimension makes the circuit less likely to have a barren plateau and easier to estimate its gradient; see \cite{ragone2024lie,heidari2024efficient} for details. In this paper, we propose to have a \emph{structural parameter sharing} scheme and apply it to additive product graphs. }

The multi-angle quantum approximate optimization algorithm (ma-QAOA) was firstly introduced in Ref.~\cite{herrman2022multi}. It introduces more classical parameters into the standard QAOA framework. Typically, $C=\sum_a C_a$ and $B=\sum_v B_v$ represent the sum of a collection of clauses. The ma-QAOA assigns each individual summand of $B$ and $C$ its own parameter. The parameters consist of $2p$ vectors: $(\boldsymbol{\gamma}_1, \boldsymbol{\gamma}_2, \ldots, \boldsymbol{\gamma}_p)$ and $(\boldsymbol{\beta}_1, \boldsymbol{\beta}_2, \ldots, \boldsymbol{\beta}_p)$.
\begin{equation}
    U(C,\boldsymbol{\gamma})=e^{-i\sum_a C_a \gamma_a}=\prod_a e^{-i\gamma_a C_a} 
\end{equation}

\begin{equation}
    U(B,\boldsymbol{\beta})=e^{-i\sum_v C_v \beta_v}=\prod_v e^{-i\beta_v B_v}.
\end{equation}

Through having more classical parameters, the ma-QAOA can enhance the performance of standard QAOA.  
In Ref.~\cite{herrman2022multi}, ma-QAOA was simulated on a collection of one-hundred triangle-free 3-regular graphs with fifty vertices and one hundred triangle-free 3-regular graphs with 100 vertices. The result shows that ma-QAOA achieved 21.26\% and 17.98\% increase in percentage on average to those of 1-QAOA respectively. 
Furthermore, the Erd\Horig{o}s–Rényi graphs with 9-12 nodes and various edge probabilities were tested for ma-QAOA with $p>1$ \cite{gaidai2024performance}, demonstrating the advantages of ma-QAOA over QAOA for multiple layers.

In our analysis of additive product graphs, we aim to leverage the advantages of the ma-QAOA while reducing the number of parameters to ensure convergence. To achieve this, we categorize the edges in the additive product graph based on their constituent atoms. If an additive product graph is constructed by $c$ atoms, then there are $c$ types of edges. Specifically, edge type $C$ comprises edges of the following form:
\begin{equation}
    \{v,vCu\},\ \{sCu,sCv\},\ \{sC^{\prime}u,sC^{\prime}uCv\}.
\end{equation}
For each type of edge, we assign a separate $\gamma$ parameter to its corresponding summand. For the mixing operator $B$, we employ a single angle parameter, consistent with the original QAOA framework. That is, in our settings, if the edge set $E$ in graph $G$ has $c$ categories, the parameters are: $(\boldsymbol{\gamma}_1, \boldsymbol{\gamma}_2, \ldots, \boldsymbol{\gamma}_p), \boldsymbol{\gamma_i} \in [0,2\pi]^c$, $(\beta_1, \beta_2, \ldots, \beta_p),\beta \in [0,2\pi]$.

\subsection{Classical local algorithms for MaxCut}
In order to better understand the performance of QAOA on MaxCut, we need to explore classical MaxCut algorithms to provide appropriate benchmarks. Since QAOA is a local algorithm, the comparison between local classical algorithms and local quantum algorithms will be reasonable. Similar to local quantum algorithm, if a classical algorithm is $k$-local, then the spin of a vertex $u$ in graph is only related to the vertices within distance $k$ of $u$. The statement implies that for the output of the algorithm, the spins of points $u$ and $v$, which are separated by a distance greater than $2k$, are independent. The best classical local algorithms known to the authors are the $k$-local algorithm proposed by Barak and Marwaha \cite{barak2022classical} and the threshold algorithm used in Refs.~\cite{hirvonen2014large,hastings2019classical,marwaha2021local}. 
There are also other classical algorithms. For example, Ref.~\cite{montanari2021optimization} addressed MaxCut on Erd\Horig{o}s–Rényi graphs, while Ref.~\cite{el2023local} extended local tree-structure graphs to local treelike structures, where a fraction of the vertices has tree-like neighborhoods. Despite the graphs under study in this paper are not high-girth graphs, we can still use existing local classical algorithms to establish good benchmarks.

The threshold algorithm was first introduced in \cite{hirvonen2014large}. Intuitively, the algorithm simply flips the spin of each vertex based on the spins of its neighbors. For a parameter $\tau\in [d/2,d]$, if a vertex has more than $\tau$ neighbors with the same spin as itself, then the spin of that vertex is flipped. The pseudocode for this algorithm is shown in \algo{threshold-algo}. 
This work studied the cut fraction of $d$-regular triangle-free graphs and derived an expected cut fraction of $1/2 + 0.28125/\sqrt{d}$ with the threshold equal to $\lceil \frac{d+\sqrt{d}}{2} \rceil$. But the threshold given in \cite{hirvonen2014large} was for notational convenience and may not be the optimal. The optimal threshold can be obtained by directly analyzing the expression for the expected cut fraction as in \thm{cut-threshold}.

\begin{theorem}[{\cite[Section 2.6]{hirvonen2014large}}]
\label{thm:cut-threshold}
    The expected cut fraction of 1-step threshold algorithm on $d$-regular triangle-free graphs is:
    \begin{equation}
        \frac12+\frac1{4^{d-1}}\binom{d-1}{\tau-1}\sum_{i=d-\tau+1}^{\tau-1}\binom{d-1}i,
    \end{equation}
    where the threshold algorithm flips the spin of a vertex if more than $\tau$ of its neighbors has the same spin.
\end{theorem}

Note that to maximize the expected cut fraction in \thm{cut-threshold}, the threshold $\tau$ needs to be optimized. 
Hastings \cite{hastings2019classical} numerically optimized the cut fraction of 1-step threshold algorithm and found it outperformed $\text{QAOA}_1$ for $3<d<1000$ except of $d=3,4,6,11$. Marwaha \cite{marwaha2021local} extended the 1-step threshold algorithm to an n-step version and directly calculated the optimal threshold $\tau_1, \tau_2$ for all $41<d<500$ by assuming $\tau_1=\tau_2$. The result showed that the 2-step threshold algorithm outperformed $\text{QAOA}_2$ for all $41<d<500$ and the performance of 2-step threshold algorithm stabilized at $1/2 + 0.417/\sqrt{d}$ for large $d$. The results were also computed for $
2 \leq d \le 50$ without assuming $\tau_1=\tau_2$, and found that 2-step threshold algorithm outperformed $\text{QAOA}_2$ when $d > 5$.

\vspace{4mm}
\begin{algorithm}[H]
    \SetAlgoLined
    \caption{Threshold Algorithm}
    \label{algo:threshold-algo}
    \KwIn{graph $G$, thresholds $\tau_1, \tau_2,\ldots, \tau_n$}
    \Begin{
    \For{every vertex $u \in G$}{
    Randomly assign a spin $+1$ or $-1$\;
    }
    \For{$i \leftarrow 1$ \KwTo $n$}{
        \For{every vertex $u \in G$}{
        \If{$u$ has the same spin as $\geq \tau_i$ of its neighbors}{
        flip the spin of $u$ after updating all vertices in this round\;
        }
        }
    }
    Cut the graph according to $+1$ or $-1$ partition\;
    \KwOut{the cut}
    }
\end{algorithm}
\vspace{4mm}

The algorithm in \cite{barak2022classical} by Barak and Marwaha achieved better bound than the threshold algorithm for 1-step and 2-step cases on high-girth regular graphs. The algorithm is done by assigning an independent random normal variable to each vertex and define the spin of each vertex according to the summation of these random variables. The pseudocode of this algorithm is shown in \algo{barak}. This algorithm achieved an expected cut fraction of at least $ 1/2+2/( \pi \sqrt {D})$ for $D$-regular graphs with girth $g > 2k+ 1$.

\begin{theorem}[{\cite[Theorem 4]{barak2022classical}}]
     For every $k$, there is a $k$-local algorithm $A$ such that for all $D$-regular $n$-vertex graphs $G$ with girth $g > 2k+ 1$, $A$ outputs a cut $x \in \{ \pm 1\} ^n $ cutting $\cos ^{- 1}( - 2\sqrt {D- 1}/ D) / \pi - O( 1/ \sqrt {k}) > 1/ 2+ 2/ ( \pi \sqrt {D}) - O( 1/ \sqrt {k})$ fraction of edges.
\end{theorem}

\vspace{4mm}
\begin{algorithm}[H]
    \SetAlgoLined
    \caption{Barak and Marwaha's Algorithm}
    \label{algo:barak}
    \KwIn{graph $G$}
    \Begin{
    \For{every vertex $w \in G$}{
    $Y_\omega \leftarrow N(0,1)$\;
    }
    \For{every vertex $u \in G$}{
    $X_u \leftarrow \mathrm{sgn}\left(\sum_{w;d(w,u)\leq k}(-1)^{d(w,u)}(D-1)^{-0.5d(w,u)}Y_w\right)$\;
    \tcp{$d(w,u)$ is the graph distance from u to w}
    }
    \KwOut{the vector $X$}
    }
\end{algorithm}
\vspace{4mm}

\section{Theoretical Guarantee of QAOA for Additive Product Graphs}\label{sec:theories}
\subsection{Classical MaxCut}
In this section, we give a theoretical guarantee for the expected cut fraction of QAOA for MaxCut on additive product graphs. Specifically, we propose an iterative formula to analyze each expectation term $\bra{\boldsymbol{\gamma},\boldsymbol{\beta}} Z_uZ_v \ket{\boldsymbol{\gamma},\boldsymbol{\beta}}$ in Eq.~\eq{expectation}. The analysis of the iterative formula is conducted within the framework of ma-QAOA discussed \sec{ma-qaoa}. 
When letting ${\gamma}_{i,1} = {\gamma}_{i,2}=\cdots={\gamma}_{i,c}$, we will have the iterative formula for the standard QAOA.

First, we need to analyze how many types of subgraphs exist for a given construction of additive product graph and what the proportion ($f_g$ in Eq.~\eq{expectation}) is for each type. 
The type of subgraphs for $\text{ma-QAOA}_p$ depends on the edge $(u,v)$ in $\bra{\boldsymbol{\gamma},\boldsymbol{\beta}} Z_uZ_v \ket{\boldsymbol{\gamma},\boldsymbol{\beta}}$ and the neighbors within distance $p$ from vertices $u$ and $v$. Based on the construction of additive product graph, every vertex $sCu$ in $V$ is connected to its neighbors following the same pattern: 
Firstly $sCu$ has neighbors $sCv$ of the same layer if $\{u,v\}$ is an edge in $\underline{A}_C$; then $sCu$ is connected to its successions $sCuC'v$ if $\{u,v\}$ is an edge in $\underline{A}_{C'}$ and its predecessors $s'C'v$ if $s = s'C'v$ and $\{u,v\}$ is an edge in $\underline{A}_C$. As a consequence, all vertices $sCu$ exhibit identical neighborhood configurations within the additive product graph regardless of choice of $s$. So the following property will hold:
\begin{proposition}
\label{prop:subgraph}
The expectation 
$\bra{\boldsymbol{\gamma},\boldsymbol{\beta}} Z_{a}Z_{b} \ket{\boldsymbol{\gamma},\boldsymbol{\beta}}$ is invariant under the following {three different forms of an edge} 
$\{a,b\}$ for any $s$ and any $C^\prime$:
\begin{equation}
    \{v,vCu\},\ \{sCu,sCv\},\ \{sC^{\prime}u,sC^{\prime}uCv\}.
\end{equation}
Consequently, each edge in generators $\underline{A}_1,\ldots,\underline{A}_c$ (in \defn{additive-graph}) corresponds to a unique isomorphic subgraph in the additive product graph. For an edge $\{u,v\}$ in $\underline{A}_C$, we use $\mathbb{E}[\underline{A}_C(u,v)]$ to denote the expectation for any $s$ and $C^\prime$ in the additive product graph.
\end{proposition}

Since the number of types of subgraph is the number of edges in $\underline{A}_1,\ldots,\underline{A}_c$, proportion $f_g$ for each type is then $\frac{1}{|E({\underline{A}_1})|+\cdots+|E({\underline{A}_c})|}$. And this yields the expectation of cut fraction \eq{expectation} to be
\begin{equation}\label{eq:cut-fraction}
    \frac{\bra{\boldsymbol{\gamma},\boldsymbol{\beta}} C_{MC} \ket{\boldsymbol{\gamma},\boldsymbol{\beta}}}{|E|}
    = \frac12 + \frac{1}{2(|E({\underline{A}_1})|+\cdots+|E({\underline{A}_c})|)} \sum_{\substack{C \in [c], \\ (a,b) \in \underline{A}_C}} 
    \mathbb{E}[\underline{A}_C(a,b)],
\end{equation}
where $E(\underline{A})$ represent the number of edges in graph $\underline{A}$.

\paragraph{Notations.}
We will introduce some notations to be used in the following. For $\text{ma-QAOA}_p$ algorithms, we denote a vector of length $2p+1$ as
\begin{equation}\label{eq:a}
    \boldsymbol{a}=(a_1,a_2,\ldots,a_p,a_0,a_{-p},\ldots,a_{-2},a_{-1}),\ a_i \in \{+1,-1\}.
\end{equation}

And for the $p$ parameter vector $(\boldsymbol{\gamma}_1, \boldsymbol{\gamma}_2,\ldots,\boldsymbol{\gamma}_p)$ in $\text{ma-QAOA}_p$, we also introduce a vectorized notation. Suppose there are $c$ atom graph building the additive product graph, then we have
\begin{equation}\label{eq:Gamma}
\begin{aligned}
    &\Gamma_1 = ({\gamma}_{1,1}, {\gamma}_{2,1},\ldots,{\gamma}_{p,1},0,-{\gamma}_{p,1},\ldots,-{\gamma}_{2,1},-{\gamma}_{1,1}) \\
    &\Gamma_2 = ({\gamma}_{1,2}, {\gamma}_{2,2},\ldots,{\gamma}_{p,2},0,-{\gamma}_{p,2},\ldots,-{\gamma}_{2,2},-{\gamma}_{1,2}) \\
    &\vdots \\
    &\Gamma_c = ({\gamma}_{1,c}, {\gamma}_{2,c},\ldots,{\gamma}_{p,c},0,-{\gamma}_{p,c},\ldots,-{\gamma}_{2,c},-{\gamma}_{1,c}).
\end{aligned}
\end{equation}

For $p$ parameter $(\beta_1, \beta_2,\ldots,\beta_p)$, we define a function of $\boldsymbol{a}$
\begin{equation}\label{eq:f}
    \begin{aligned}
    f\left( \boldsymbol{a} \right) &= \frac12\bra{ a_1} e^{i\beta_1X} \ket{a_2} \cdots 
    \bra{a_{p-1}} e^{i\beta_{p-1} X} \ket{a_p} 
    \bra{a_p} e^{i\beta_p X} \ket{a_0} \\
    & \ \ \times \bra{a_0} e^{-i\beta_p X} \ket{a_{-p}}
    \bra{a_{-p}} e^{-i\beta_{p-1} X} \ket{a_{-(p-1)}} \cdots 
    \bra{a_{-2}} e^{-i\beta_1 X} \ket{a_{-1}} 
    \end{aligned}.
\end{equation}

\paragraph{Iterations for infinite additive product graph \texorpdfstring{$X$}{X}.}
Let us consider $p=2$ and take the upper graph in \fig{additive-graph} for example. There are technically six kinds of expectation to be computed: 
\begin{align}
\mathbb{E}[\underline{A}_1(u,v)],\ \mathbb{E}[\underline{A}_1(v,w)],\ \mathbb{E}[\underline{A}_1(w,u)],\ \mathbb{E}[\underline{A}_2(u,v)],\ \mathbb{E}[\underline{A}_3(v,w)],\ \mathbb{E}[\underline{A}_4(u,w)]. 
\end{align}
Here, we give the iterative formula of $\mathbb{E}[\underline{A}_1(v,w)]$ as an example.

\begin{figure}[H]
    \centering
    \includegraphics[scale=0.6]{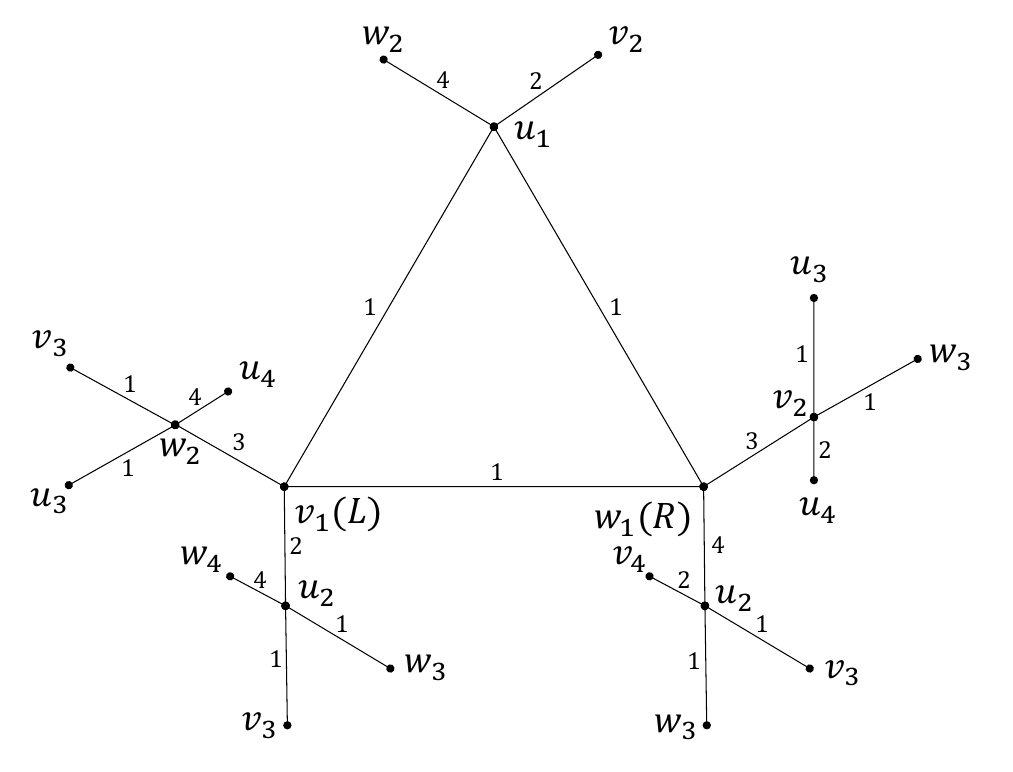}
    \caption{The subgraph in additive product graph corresponding to the edge $\{v,w\}$ in $\underline{A}_1$ at $p=2$. The labels of vertices and edges come from the generation of additive product graph. 
    }
    \label{fig:subgraph}
\end{figure}

The subgraph associated with $\mathbb{E}[\underline{A}_1(v,w)]$ is shown in \fig{subgraph}. This subgraph has $n=21$ nodes, and we restrict the state $\ket{\boldsymbol{\gamma},\boldsymbol{\beta}}$ on these $n$ nodes. To avoid confusion in the name of vertices, we denote the edge $(v_1,w_1)$ we are seeking the expectation for as $(L,R)$. In the following, we evaluate the expectation of $\bra{\boldsymbol{\gamma},\boldsymbol{\beta}} Z_LZ_R \ket{\boldsymbol{\gamma},\boldsymbol{\beta}}$. Following the line of computation in \cite{basso2022quantum}, we insert 5 complete sets $\{\boldsymbol{z}^{[1]},\boldsymbol{z}^{[2]},\boldsymbol{z}^{[0]},\boldsymbol{z}^{[-2]},\boldsymbol{z}^{[-1]}\}$ in the computational Z-basis, with each $z^{[i]}$ ranging over all $2^n$ elements of $\{+1,-1\}^n$. 
\begin{equation}\label{eq:expec}
\begin{aligned}
\bra{\boldsymbol{\gamma},\boldsymbol{\beta}} Z_LZ_R \ket{\boldsymbol{\gamma},\boldsymbol{\beta}}
    &=\sum_{ \{\boldsymbol{z}^{[i]}\} } \braket{s|\boldsymbol{z}^{[1]}} 
    U(C(\boldsymbol{z}^{[1]}),\boldsymbol{\gamma}_1)
    \bra{\boldsymbol{z}^{[1]}} e^{i\beta_1 B} \ket{\boldsymbol{z}^{[2]}} 
    U(C(\boldsymbol{z}^{[2]}),\boldsymbol{\gamma}_2)
   \bra{\boldsymbol{z}^{[2]}} e^{i\beta_2 B} \ket{\boldsymbol{z}^{[0]}} z_L^{[0]}z_R^{[0]} \\
    &\times \bra{\boldsymbol{z}^{[0]}} e^{-i\beta_2 B} \ket{\boldsymbol{z}^{[-2]}}
    U(C(\boldsymbol{z}^{[2]}),-\boldsymbol{\gamma}_2)
    \bra{\boldsymbol{z}^{[-2]}} e^{-i\beta_1 B} \ket{\boldsymbol{z}^{[-1]}}
    U(C(\boldsymbol{z}^{[1]}),-\boldsymbol{\gamma}_1)
    \braket{\boldsymbol{z}^{[-1]}|s}\\
    &=\frac{1}{2^n} \sum_{\{\boldsymbol{z}^{[i]}\}}
    U(C(\boldsymbol{z}^{[1]}),\boldsymbol{\gamma}_1)
    U(C(\boldsymbol{z}^{[2]}),\boldsymbol{\gamma}_2)
    U(C(\boldsymbol{z}^{[2]}),\boldsymbol{-\gamma}_2)
    U(C(\boldsymbol{z}^{[1]}),\boldsymbol{-\gamma}_1)
    z_L^{[0]}z_R^{[0]} \\
    &\times \prod_{k=1}^n \bra{ z_k^{[1]}} e^{i\beta_1 X}\ket{z_k^{[2]}}
   \bra{z_k^{[2]}} e^{i\beta_2 X} \ket{z_k^{[0]}}
   \bra{z_k^{[0]}} e^{-i\beta_2 X} \ket{ z_k^{[-2]}} 
   \bra{z_k^{[-2]}} e^{-i\beta_1 X} \ket{z_k^{[-1]}},
\end{aligned}
\end{equation}
where 
\begin{equation}
\ket{s} = \ket{+}^{\otimes n},\quad
    U(C(\boldsymbol{z}^{[i]}),\boldsymbol{\gamma}_i) =  
    \exp\left(i \sum_{(k_1,k_2) \in E} \boldsymbol{\gamma}_{i,c(k_1,k_2)} (-\boldsymbol{z}^{[i]}_{k_1} \boldsymbol{z}^{[i]}_{k_2})\right) ,
\end{equation}
with the subscript $c_{(k_1,k_2)}$ of $\boldsymbol{\gamma}_{i,c(k_1,k_2)}$ represent the category of edge $(k_1,k_2)$.

\newpage
We can leverage Eqs.~\eq{a}, \eq{Gamma}, and \eq{f} to reformulate Eq.~\eq{expec}. Instead of summing over $\{\boldsymbol{z}^{[1]},\boldsymbol{z}^{[2]},\boldsymbol{z}^{[0]},\\\boldsymbol{z}^{[-2]},\boldsymbol{z}^{[-1]}\}$, we can switch the summation to the configuration basis. Consider the node $k$ within the basis vectors $\{\boldsymbol{z}^{[1]},\boldsymbol{z}^{[2]},\boldsymbol{z}^{[0]},\boldsymbol{z}^{[-2]},\boldsymbol{z}^{[-1]}\}$. we can represent $\{\boldsymbol{z}^{[1]}_k,\boldsymbol{z}^{[2]}_k,\boldsymbol{z}^{[0]}_k,\boldsymbol{z}^{[-2]}_k,\boldsymbol{z}^{[-1]}_k\}$ using the vector $\boldsymbol{z}_k$. And the configuration basis is then given by $\{\boldsymbol{z}_k,\ 1 \leq k \leq n\}$, where each $\boldsymbol{z}_k$ runs through all $2^{2p+1}$ elements of $\{+1,-1\}^{2p+1}$.
\begin{equation}\label{eq:expec-config-basis}
    \bra{\boldsymbol{\gamma},\boldsymbol{\beta}} Z_LZ_R \ket{\boldsymbol{\gamma},\boldsymbol{\beta}}
    =\sum_{\{ \boldsymbol{z}_k \}} z_L^{[0]} z_R^{[0]}
    \exp\left[
    -i \sum_{ (k_1,k_2)\in E}
    \boldsymbol{\Gamma}_{c_{(k_1,k_2)}} \cdot (\boldsymbol{z}_{k_1}\boldsymbol{z}_{k_2})\right] \prod_{k=1}^n f(\boldsymbol{z}_k),
\end{equation}
where the subscript $c_{(k_1,k_2)}$ of $\boldsymbol{\Gamma}$ represent the category of edge $(k_1,k_2)$ and $\boldsymbol{z}_{k_1}\boldsymbol{z}_{k_2}$ represents element-wise product.

To calculate Eq.~\eq{expec-config-basis}, we need to sum over the configuration basis over each node in the subgraph in \fig{subgraph}. We can firstly sum over the leaf node $u_4$ attached to the node $w_2$. 

The summation over the configuration basis $\boldsymbol{z}_{u_4}$ yields
\begin{equation}\label{eq:sum-leaf1}
    \sum_{\boldsymbol{z}_{u_4}} f(\boldsymbol{z}_{u_4}) \exp(-i\Gamma_4 \cdot \boldsymbol{z}_{u_4} \boldsymbol{z}_{w_2}).
\end{equation}

Eq.~\eq{sum-leaf1} can also be interpreted as the contributions from the succession nodes $s3w4k,\ k \in \underline{A}_4$ to node $s3w$. We can sum the contribution of $s3w4k,\ k \in \underline{A}_4$ to $s3w$ using an iterative formula $G_{4,w}^{[m]}(z_w)$. The superscript $m$ denote how far the node $s3w$ can see. As is known the ma-QAOA is a local algorithm and the summation is done on a local subgraph. So the superscirpt $m$ means the node $s3w$ can get the contribution of $s3w4k,\ k \in \underline{A}_4$ only if the distance between node $w$ and node $k$ in graph $\underline{A}_4$ is no more than $m$. For the node $w_2$, the superscript $m=1$ and we can reformulate Eq.~\eq{sum-leaf1} as
\begin{equation}\label{eq:sum-leaf1-G}
   G_{4,w}^{[1]}(\boldsymbol{z}_{w_2}) = 
   \sum_{ \substack{ \{\boldsymbol{z}_k\} \\ k \in \underline{A}_4,\ \text{dis}(k,w) \leq 1 } }
   \exp\left[ -i\Gamma_4 \cdot \left( \sum_{\substack{(k_1,k_2) \in E(\underline{A}_4)\\
   \text{dis}((k_1,k_2),w) \leq 1} } z_{k_1}z_{k_2} \right)\right] 
   \prod_{k \neq w} f(z_k).
\end{equation}

Summing over the configuration basis of nodes $v_3$ and $u_3$ to $w_2$ in \fig{subgraph} follows the same pattern. We need to sum the configuration basis of the successions generated by graph $\underline{A}_1$. And the intuition of the iterative formula $G_{1,w}^{[1]}(z_w)$ is the same. 
\begin{equation}\label{eq:sum-leaf2-G}
    G_{1,w}^{[1]}(\boldsymbol{z}_{w_2}) = 
   \sum_{ \substack{ \{\boldsymbol{z}_k\} \\ k \in \underline{A}_1,\ \text{dis}(k,w) \leq 1 } }
   \exp\left[ -i\Gamma_1 \cdot \left( \sum_{\substack{(k_1,k_2) \in E(\underline{A}_1)\\
   \text{dis}((k_1,k_2),w) \leq 1} } z_{k_1}z_{k_2} \right)\right] 
   \prod_{k \neq w} f(z_k).
\end{equation}

After performing the sums for each succession $sCwC^\prime v$ of $sCw$, we can simply multiply the iterative formula $G_{C^\prime,w}^{[m]}(\boldsymbol{z}_w)$ together to obtain the total contribution form all successor nodes $sCwC^\prime v$, where $C^\prime \neq C$ and $v \in \underline{A}_{C^\prime}$
\begin{equation}
    \prod_{C^\prime \neq C} G_{C^\prime,w}^{[m]}(\boldsymbol{z}_w).
\end{equation}

Again, summing nodes $w_2$ and $u_2$ to the parent node $v_1$ in \fig{subgraph} yields 
\begin{equation}
    \sum_{\boldsymbol{z}_{w_2},\boldsymbol{z}_{u_2}} f(\boldsymbol{z}_{w_2}) f(\boldsymbol{z}_{u_2})  G_{1,w}^{[1]}(\boldsymbol{z}_{w_2})  G_{1,u}^{[1]}(\boldsymbol{z}_{u_2}) \exp(-i\left(
    \Gamma_3 \cdot \boldsymbol{z}_{w_2} \boldsymbol{z}_{v_1}+
    \Gamma_2 \cdot \boldsymbol{z}_{u_2} \boldsymbol{z}_{v_1} \right)), 
\end{equation}
which can also be written as an iteration
\begin{equation}
\begin{aligned}
    \prod_{C^\prime \neq C} G_{C^\prime,v}^{[2]}(\boldsymbol{z}_{v_1})
    &=G_{3,v}^{[2]}(\boldsymbol{z}_{v_1})G_{2,v}^{[2]}(\boldsymbol{z}_{v_1})\\
    &=\sum_{\boldsymbol{z}_{w_2},\boldsymbol{z}_{u_2}} f(\boldsymbol{z}_{w_2}) f(\boldsymbol{z}_{u_2})  G_{1,w}^{[1]}(\boldsymbol{z}_{w_2})  G_{1,u}^{[1]}(\boldsymbol{z}_{u_2}) \exp(-i\left(
    \Gamma_3 \cdot \boldsymbol{z}_{w_2} \boldsymbol{z}_{v_1}+
    \Gamma_2 \cdot \boldsymbol{z}_{u_2} \boldsymbol{z}_{v_1} \right))\\
    &=\sum_{ \substack{ \{\boldsymbol{z}_k\} \\ k \in \underline{A}_3,\ \text{dis}(k,v) \leq 1 } }
    \prod_{k \neq w} \left[ f(z_k) \prod_{C \neq 3} G_{C,k}^{[2-\text{dis}(k,v)]}(\boldsymbol{z}_k) \right]
   \exp\left[ -i\Gamma_3 \cdot \left( \sum_{\substack{(k_1,k_2) \in E(\underline{A}_3)\\
   \text{dis}((k_1,k_2),v) \leq 1} } z_{k_1}z_{k_2} \right)\right] 
    \\
   &\times \sum_{ \substack{ \{\boldsymbol{z}_k\} \\ k \in \underline{A}_2,\ \text{dis}(k,v) \leq 1 } }
    \prod_{k \neq w} \left[ f(z_k) \prod_{C \neq 2} G_{C,k}^{[2-\text{dis}(k,v)]}(\boldsymbol{z}_k) \right]
   \exp\left[ -i\Gamma_2 \cdot \left( \sum_{\substack{(k_1,k_2) \in E(\underline{A}_2)\\
   \text{dis}((k_1,k_2),v) \leq 1} } z_{k_1}z_{k_2} \right)\right] .
\end{aligned}
\end{equation}

The iterative formula for the remaining vertices can be derived in a similar fashion. Finally, we sum the configuration basis corresponding to the three vertices in the triangle containing the edge $(L,R)$. This yields the $p=2$ ma-QAOA expectation of edge $(L,R)$ to 
\begin{align}
    \bra{\boldsymbol{\gamma},\boldsymbol{\beta}} Z_LZ_R \ket{\boldsymbol{\gamma},\boldsymbol{\beta}}&= 
        \sum_{\boldsymbol{z}_{v_1},\boldsymbol{z}_{u_1},\boldsymbol{z}_{w_1}}
        \boldsymbol{z}_{v_1}^{[0]}\boldsymbol{z}_{w_1}^{[0]}
        f(\boldsymbol{z}_{v_1})f(\boldsymbol{z}_{w_1})f(\boldsymbol{z}_{u_1})\nonumber \\
        &\quad\quad \times G_{2,v}^{[2]}(\boldsymbol{z}_{v_1})G_{3,v}^{[2]}(\boldsymbol{z}_{v_1}) G_{3,w}^{[2]}(\boldsymbol{z}_{w_1})G_{4,w}^{[2]}(\boldsymbol{z}_{w_1})
        G_{2,u}^{[1]}(\boldsymbol{z}_{u_1})G_{4,u}^{[1]}(\boldsymbol{z}_{u_1})\nonumber \\
        &\quad\quad \times \exp(-i\Gamma_1 \cdot \left(\boldsymbol{z}_{w_1} \boldsymbol{z}_{v_1}+\boldsymbol{z}_{u_1} \boldsymbol{z}_{v_1} + \boldsymbol{z}_{u_1} \boldsymbol{z}_{w_1} \right))\nonumber \\
        & = \sum_{ \substack{ \{\boldsymbol{z}_k\} \\ k \in \underline{A}_1,\ \min(\text{dis}(k,L),\text{dis}(k,R) )\leq 2 } }
        z_L^{[0]}z_R^{[0]}
        \prod_{k} \left[ f(z_k) \prod_{C \neq 1} G_{C,k}^{[2-\min(\text{dis}(k,L),\text{dis}(k,R)]}(\boldsymbol{z}_k) \right]\nonumber \\
        &\quad\quad \times \exp\left[ -i\Gamma_1 \cdot \left( \sum_{\substack{(k_1,k_2) \in E(\underline{A}_1)\\
   \min(\text{dis}((k_1,k_2),L),\text{dis}((k_1,k_2),R)) \leq 2}} z_{k_1}z_{k_2} \right)\right].
\end{align}

For general graphs and higher $p$, we simply need to calculate the expectation for each type of subgraph. As proposed in \prop{subgraph} and \eq{cut-fraction}, the expectation of subgraphs in an additive product graph depends only on the choice of $C$ and $(u, v) \in \underline{A}_C$. For higher $p$, $C \in [c]$ and $(a, b) \in \underline{A}_C$, every $\mathbb{E}[\underline{A}_C(a, b)]$ can be calculated iteratively, leveraging the same approach. Specifically, the $m$-depth iteration on vertex $a$ in the atom graph $C$ can be calculated using the iterative formula $G_{C,a}^{[m]}$. By incorporating contributions from the child nodes via this formula, the expectation can be obtained. In the following, we present a formal theorem to illustrate the iterative formula for ma-QAOA on additive product graphs.
\begin{theorem}[expected cut fraction on additive product graph]
    Let $G=A_1\oplus \cdots \oplus A_c$ be an additive product graph. For subgraph corresponding to each edge $(a, b) \in \underline{A}_C,\ C \in [c]$, the expectation $\mathbb{E}[\underline{A}_C(a, b)]$ is given by:
    \begin{equation}
    \begin{aligned}\label{eq:EA}
        \mathbb{E}[\underline{A}_C(a,b)]&=-\sum_{ \substack{ \{\boldsymbol{z}_k\} \\ k \in \underline{A}_C,\ \min(\text{dis}(k,a),\text{dis}(k,b) )\leq p } }
        z_a^{[0]}z_b^{[0]}
        \prod_{k} \left[ f(z_k) \prod_{C^\prime \neq C} G_{C^\prime,k}^{[p-\min(\text{dis}(k,a),\text{dis}(k,b)]}(\boldsymbol{z}_k) \right] \\
        &\quad\quad \times \exp\left[ -i\Gamma_C \cdot \left( \sum_{\substack{(k_1,k_2) \in E(\underline{A}_C)\\
        \min(\text{dis}((k_1,k_2),L),\text{dis}((k_1,k_2),R)) \leq p} } z_{k_1}z_{k_2} \right)\right] ,
    \end{aligned}
    \end{equation}
    where $G_{C,a}^{[0]}(\boldsymbol{z}_a)=1$, $G_{C,a}^{[m]}(\boldsymbol{z}_a)$ is defined for each $C \in [c]$, $a \in E(\underline{A}_C)$, and $m = 1, 2, \ldots, p$ as follows:
    \begin{equation}\label{eq:G}
        G_{C,a}^{[m]}(\boldsymbol{z}_a)=\sum_{ \substack{ \{\boldsymbol{z}_k\} \\ k \in \underline{A}_C,\ \text{dis}(k,a) \leq m } }
        \prod_{k \neq a} \left[ f(z_k) \prod_{C^\prime \neq C} G_{C^\prime,k}^{[m-\text{dis}(k,a)]}(\boldsymbol{z}_k) \right]
        \exp\left[ -i\Gamma_C \cdot \left( \sum_{\substack{(k_1,k_2) \in E(\underline{A}_C) \\
        \text{dis}((k_1,k_2),a) \leq m}} z_{k_1}z_{k_2} \right)\right].
    \end{equation}
    
    Finally, the expected cut fraction of ma-QAOA on the additive product graph $G$ is obtained by summing the expectations $\mathbb{E}[\underline{A}_C(a, b)]$ for each subgraph.
    \begin{equation}
        \frac{\bra{\boldsymbol{\gamma},\boldsymbol{\beta}} C_{MC} \ket{\boldsymbol{\gamma},\boldsymbol{\beta}}}{|E|}
        = \frac12 + \frac{1}{2(|E({\underline{A}_1})|+\cdots+|E({\underline{A}_c})|)} \sum_{\substack{C \in [c], \\ (a,b) \in \underline{A}_C}} 
        \mathbb{E}[\underline{A}_C(a,b)].
    \end{equation}
\end{theorem}

The iterative formula can be further simplified. 
The upper graph in \fig{additive-graph} exhibits symmetry between its atoms, allowing the edges to be categorized into two groups instead of four. Specifically, one group comprises the edges within the triangles, while the other group consists of the cut edges connecting the triangles. The application of this symmetry property reduces the number of parameters but can still describe the edge structure of the graph. Consequently, the iterative formula for this graph is simplified, yielding $G_{C,v}^{[m]}=G_{C,u}^{[m]}=G_{C,w}^{[m]}$ for all $C\in [c]$ and $m \in [0,p]$, provided that the iterative formula exists.

\subsection{Generalizing classical MaxCut to quantum MaxCut}

The quantum MaxCut problem can be seen as a natural quantum extension of the classical MaxCut problem, and it is defined as follows: 
\begin{definition}
    Given a graph $G = (V,E)$, the goal of the quantum MaxCut problem is to determine the largest eigenvalue of the following Hamiltonian:
    \begin{equation}
        H = \sum_{(u,v) \in E} \frac12 (1-X_uX_v-Y_uY_v-Z_uZ_v).
    \end{equation}
\end{definition}
As a comparison, the cost function operator for the classical MaxCut problem is given by
$C_{MC} = \sum_{(u,v) \in E} \frac{1}{2} \big(-Z_u Z_v + 1\big), $
as defined in \eq{maxcut-operator}. The quantum MaxCut problem generalizes the classical version by incorporating Pauli $X$ and $Y$ operators on each edge, thereby capturing more quantum correlations beyond the classical MaxCut problem.

Ref.~\cite{kannan2024quantum} investigated iterative formulas for the quantum MaxCut problem on high-girth regular graphs within the framework of the Hamiltonian QAOA. In their work, the QAOA state is prepared using four driver Hamiltonians: $A= \sum_{(u,v) \in E} Z_u Z_v,\ B=\sum_{v \in V}X_v,\ C=\sum_{v \in V}Z_v,\ D=\sum_{v \in V}\boldsymbol{n}_v \cdot (X_v,Y_v,Z_v)$, along with four parameters $\boldsymbol{\Theta} = (\boldsymbol{\alpha},\boldsymbol{\beta},\boldsymbol{\gamma},\boldsymbol{\delta}) \in \mathbb{R}^{4p}$. Here $\boldsymbol{n} = \{\boldsymbol{n}_v\}_{v \in V}$ and $ \boldsymbol{m} = \{\boldsymbol{m}_v\}_{v \in V}$, where each $\boldsymbol{n}_v$ and $\boldsymbol{m}_v$ lie on the 3-dimensional unit sphere; $\boldsymbol{n}_v$ gives the direction of the driving Hamiltonian $D$ and $\ket{\boldsymbol{m}_v}$ is the starting state. 
The QAOA state is then defined as:
\begin{equation}
   \left|\boldsymbol{\Theta}, G, \boldsymbol{n}, \boldsymbol{m}\right\>
   = 
   e^{-i\delta_p D} e^{-i\gamma_p C} e^{-i\beta_p B} e^{-i\alpha_p A} \cdots e^{-i\delta_1 D} e^{-i\gamma_1 C} e^{-i\beta_1 B} e^{-i\alpha_1 A} \left( \otimes_{v\in V} \ket{\boldsymbol{m}_v} \right),
\end{equation}
and the loss function is given by the expectation value
\begin{equation}
    \mathbb{E}_{\boldsymbol{m},\boldsymbol{n}}\left[ 
    \left\<\boldsymbol{\Theta}, G, \boldsymbol{n}, \boldsymbol{m}\right|
    \frac12 (1-X_uX_v-Y_uY_v-Z_uZ_v)
    \left|\boldsymbol{\Theta}, G, \boldsymbol{n}, \boldsymbol{m}\right\>
    \right]
\end{equation}
which can be computed by an iterative formula in~\cite{kannan2024quantum}, which can be understood as the Hamiltonian for a single edge under the prepared QAOA state on high-girth regular graphs.

Our analysis of QAOA's iterative formula for MaxCut on additive product graphs naturally extends to the quantum MaxCut problem. Specifically, the driver Hamiltonians $B$,$C$, and $D$ act on individual vertices, while the Hamiltonian $A$ operates on pairs of vertices connected by an edge. This allows us to reuse the iteration technique: summing over configurations at the leaf nodes and propagating contributions to parent nodes, with only minor modifications to the iteration terms. Similarly,
we define the vector associated with node $v$ as $\boldsymbol{z}_v$ of length $2p+2$ as 
\begin{equation*}
 \boldsymbol{z}_v=(z_v^{[1]},z_v^{[2]},\ldots,z_v^{[p+1]},z_v^{[-(p+1)]},\ldots,z_v^{[-2]},z_v^{[1-]}),\ z_v^{[i]} \in \{+1,-1\}
\end{equation*}
and 
\begin{equation*}
    \boldsymbol{\mathcal{A}} =(\alpha_1,\alpha_2,\ldots,\alpha_p,0,0,-\alpha_{p},\ldots,-\alpha_{2},-\alpha_{1}).
\end{equation*}
We can then state the following corollary, which calculates the iterative formula for quantum MaxCut on additive product graphs.
\begin{corollary}
\label{cor:quantum-maxcut}
Let $G=A_1\oplus \cdots \oplus A_c$ be an additive product graph. For subgraph corresponding to each edge $(a, b) \in \underline{A}_C,\ C \in [c]$, the expectation of quantum MaxCut on this edge, where the two vertices on the edge are denoted as $L$ and $R$, is given by:
\begin{equation}
    \mathbb{E}_{QMC}[\underline{A}_C(a,b)] = \mathbb{E}_{\boldsymbol{m},\boldsymbol{n}}\left[ 
    \left\<\boldsymbol{\Theta}, G, \boldsymbol{n}, \boldsymbol{m}\right|
    \frac12 (1-X_uX_v-Y_uY_v-Z_uZ_v)
    \left|\boldsymbol{\Theta}, G, \boldsymbol{n}, \boldsymbol{m}\right\>
    \right].
\end{equation}

Let function $f^{\sigma}(\boldsymbol{z}_v)$ denote the summation with respect to each vertex,
\begin{equation}
\label{eq:f_new}
    f^{\sigma}(\boldsymbol{z}_v)  = \mathbb{E}_{\boldsymbol{m}_v,\boldsymbol{n}_v}\left[ 
\begin{aligned}
    \braket{\boldsymbol{m}_v | z_v^{[1]}} 
    \bra{z_v^{[1]}} e^{i\beta_1 B}e^{i\gamma_1 C}e^{i\delta_1 D}   \ket{z_v^{[2]}} 
    \ldots
   \bra{z_v^{[p]}} e^{i\beta_p B}e^{i\gamma_p C}e^{i\delta_p D}   \ket{z_v^{[p+1]}}
   \bra{z_v^{[p+1]}} \sigma \ket{z_v^{[-(p+1)]}} \\
    \bra{z_v^{[-(p+1)]}} e^{-i\delta_p D}e^{-i\gamma_p C} e^{-i\beta_p B}  \ket{z_v^{[-p]}}
    \ldots
    \bra{z_v^{[-2]}} e^{-i\delta_1 D}e^{-i\gamma_1 C} e^{-i\beta_1 B}  \ket{z_v^{[-1]}} \braket{z_v^{[-1]} | \boldsymbol{m}_v } 
\end{aligned}
    \right].
\end{equation}
The iterative formula for $G_{C,a}^{[m]}(\boldsymbol{z}_a)$, which represents the summation over the configuration basis of the successive nodes of node $a$, is given by:
\begin{equation}
\label{eq:G_new}
    G_{C,a}^{[m]}(\boldsymbol{z}_a)=\sum_{ \substack{ \{\boldsymbol{z}_k\} \\ k \in \underline{A}_C,\ \text{dis}(k,a) \leq m } }
    \prod_{k \neq a} \left[ f^{I}(z_k) \prod_{C^\prime \neq C} G_{C^\prime,k}^{[m-\text{dis}(k,a)]}(\boldsymbol{z}_k) \right]
    \exp\left[ -i \mathcal{A} \cdot \left( \sum_{\substack{(k_1,k_2) \in E(\underline{A}_C) \\
    \text{dis}((k_1,k_2),a) \leq m}} z_{k_1}z_{k_2} \right)\right].
\end{equation}

Finally, the expectation for $\sigma  = \frac12 (1-X_LX_R-Y_LY_R-Z_LZ_R) $ is
\begin{equation}
\label{eq:EA_new}
\begin{aligned}
    &\mathbb{E}_{\boldsymbol{m},\boldsymbol{n}}\left[ \bra{ \boldsymbol{\Theta}, G, \boldsymbol{n}, \boldsymbol{m}} \sigma \ket{ \boldsymbol{\Theta}, G, \boldsymbol{n}, \boldsymbol{m}}\right] \\
    = &-\sum_{ \substack{ \{\boldsymbol{z}_k\} \\ k \in \underline{A}_C,\ \min(\text{dis}(k,a),\text{dis}(k,b) )\leq p } }
    \prod_{k } \left[ f^{\tilde{\sigma}}(z_k) \prod_{C^\prime \neq C} G_{C^\prime,k}^{[p-\min(\text{dis}(k,a),\text{dis}(k,b)]}(\boldsymbol{z}_k) \right] \\
    &\quad\quad \times \exp\left[ -i \mathcal{A} \cdot \left( \sum_{\substack{(k_1,k_2) \in E(\underline{A}_C)\\
    \min(\text{dis}((k_1,k_2),L),\text{dis}((k_1,k_2),R)) \leq p} } z_{k_1}z_{k_2} \right)\right] ,
\end{aligned}
\end{equation}
where $\tilde{\sigma} = \sigma $ if $k \in \{L,R\}$, and $\tilde{\sigma} = I $ otherwise. 
\end{corollary}

\begin{proof}[Proof sketch]
For both classical MaxCut and quantum MaxCut on additive product graphs, we adopt a similar methodology by analyzing iterative formulas on these graphs. This approach involves propagating contributions from child nodes to parent nodes.
We begin by defining a vector $\boldsymbol{z}_v$, associated with each node $v$. The length of $\boldsymbol{z}_v$ is $2p+1$ in the analysis of MaxCut.  In the quantum MaxCut problem, the length of $\boldsymbol{z}_v$ is $2p+2$ since an additional computational basis is required to account for the fact that $\frac12 (1-X_LX_R-Y_LY_R-Z_LZ_R)$ is not diagonal, in contrast to $\frac12 (1-Z_LZ_R)$. Similarly, we define the parameter vector $\mathcal{A}$. 
The function in \eq{f_new} is analogous to that in \eq{f} applied to the classical MaxCut setting, as it computes the product associated with each vertex $v$. Likewise, the iterative formula \eq{G_new} corresponds to \eq{G}, with Eq.~\eq{f_new} replaces \eq{f} and the parameter $\mathcal{A}$ substitutes $\Gamma_C$. Both formulas utilize the same method to analyze the nodes in the additive product graph. Finally, we can obtain the expectation of the iterative formula \eq{EA_new} with \eq{f_new}, \eq{G_new}, and parameter $\mathcal{A}$, which can be considered as a generalization of \eq{EA}.
\end{proof}

\section{Experiments}
\label{sec:experiments}

\subsection{Comparison with classical local algorithms}
In this section, we will perform numerical evaluations on several cases of additive product graphs in \fig{graph-set1}. First, we compare the results of different classical local algorithms for MaxCut known to the authors.  Specifically, we will test the threshold algorithm as well as the variations of Barak and Marwaha's $k$-local algorithm to provide convincing benchmarks on the graphs tested. 

\label{sec:comp-classical}
\begin{figure}[H]
\centering
    \subcaptionbox{\label{fig:a}}{\includegraphics[width = 0.3\textwidth, page={1}]{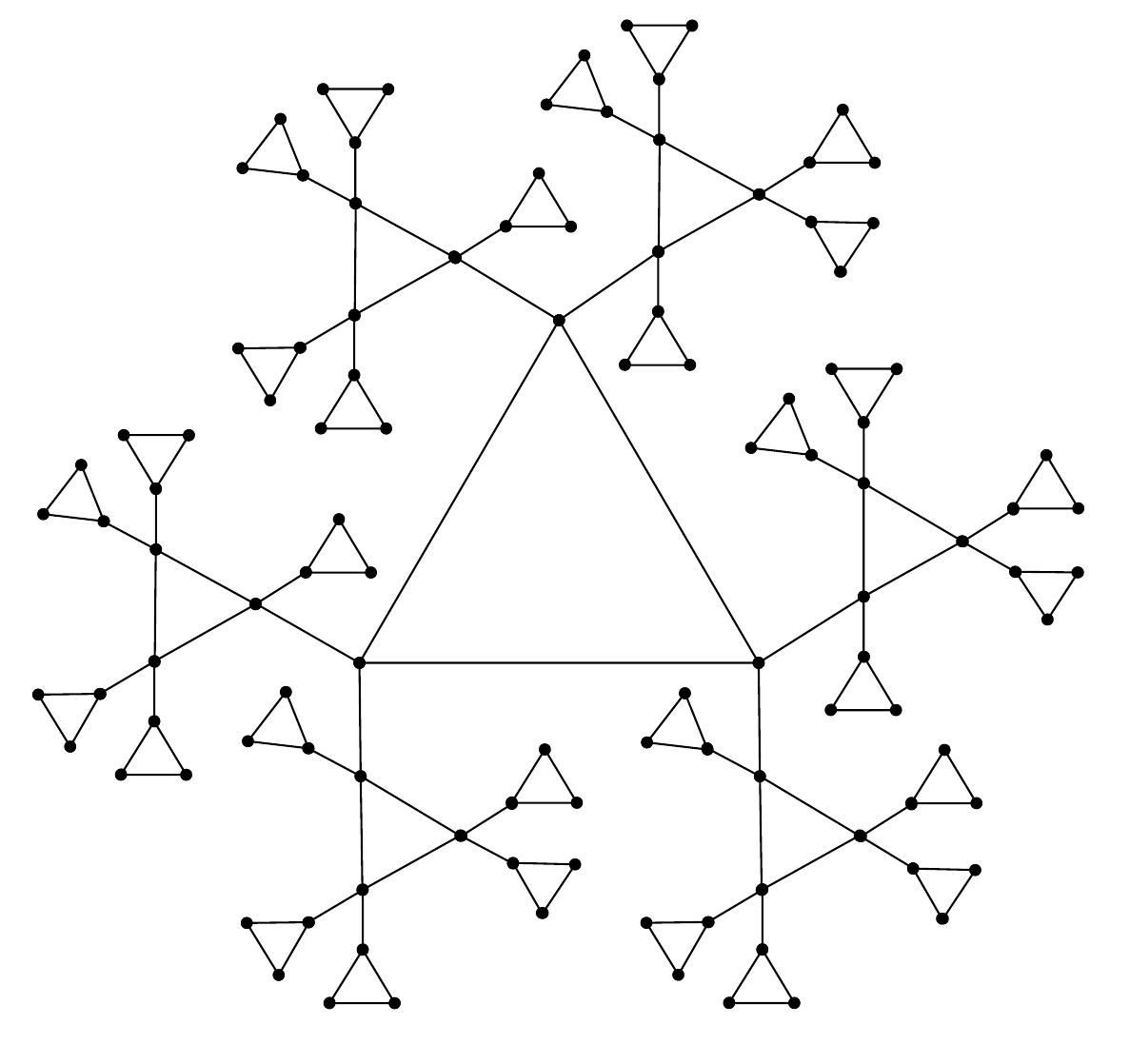}}
    \hspace{6mm}
    \subcaptionbox{\label{fig:b}}{\includegraphics[width = 0.3\textwidth, page={2}]{exp_graph.pdf}}
    \hspace{6mm}
    \subcaptionbox{\label{fig:c}}{\includegraphics[width = 0.3\textwidth, page={3}]{exp_graph.pdf}} 
\caption{The additive product graphs used in the QAOA and classical algorithm experiments.}
\label{fig:graph-set1}
\end{figure}

The threshold algorithm was firstly introduced in \cite{hirvonen2014large}. Hastings \cite{hastings2019classical} and Marwaha \cite{marwaha2021local} compared threshold algorithms with QAOA on high-girth regular graphs for $p=1,2$ respectively. Furthermore, our numerical experiments show that the threshold algorithm can provide convincing benchmarks against QAOA on Ramanujan graphs. Since the graphs in \fig{graph-set1} are no longer high-girth graphs, it is hard to derive a universal bound for the cut fraction of threshold algorithm. So Monte Carlo algorithm is used to calculate the expected cut fraction of the threshold algorithm in \algo{threshold-algo}. We obtained stable results (see \tab{classical-comp}) by repeating the experiments $10^6$ times for 1-step and 2-step version of threshold algorithm in one graph, and $10^5$ times for the 3-step threshold algorithm. The optimal threshold (see in \tab{para-threshold}) in the threshold algorithm was obtained through brute force search of all possible parameter values.

\vspace{4mm}
\begin{algorithm}[H]
    \SetAlgoLined
    \caption{Variations of Barak and Marwaha's Algorithm}
    \label{algo:barak_vari}
    \KwIn{graph $G$}
    \Begin{
    \For{every vertex $w \in G$}{
    $Y_\omega \leftarrow N(0,1)$\;
    }
    \For{every vertex $u \in G$}{
    \tcp{Para1: Use the original choice of parameters. }
    $X_u^1 \leftarrow \mathrm{sgn}\left(\sum_{w;d(w,u)\leq k}(-1)^{d(w,u)}(D-1)^{-0.5d(w,u)}Y_w\right)$\;
    
    \vspace{1em}
    \tcp{Para2: Let $S=\{w|d(w,u)=k\}$, the vertices of distance $d(w,u)$ from vertex $u$ are normalized. }
    {$S\leftarrow \{w|d(w,u)=k\}$;}
    
    $X_u^2 \leftarrow \mathrm{sgn}\left(\sum_{w;d(w,u)\leq k}(-1)^{d(w,u)} {|S|}^{-0.5}Y_w\right)$\;
    
    \vspace{1em}
    \tcp{Para3: For each vertex $u$ and its neighbors $v$, define $S_v=\{w|d(w,u)=k \wedge d(w,u) < d(w,v) \}$, representing the vertices at distance $d(w,u)$ from vertex $u$ that are closer to $u$ than to $v$. Note that different neighbors $v$ of $u$ result in different $|S_v|$. The algorithm evaluates all possible choice of parameter $|S_v|$. }
    {$S_v\leftarrow \{w|d(w,u)=k \wedge d(w,u) < d(w,v) \}$}
    
    $X_u^3 \leftarrow \mathrm{sgn}\left(\sum_{w;d(w,u)\leq k}(-1)^{d(w,u)} {|S_v|}^{-0.5}Y_w\right)$\;

    \vspace{1em}
    \tcp{Para4: For each vertex $u$ and its neighbors $v$, define ${S_v}^\prime=\{w|d(w,u)=k \wedge d(w,u) \leq d(w,v) \}$, representing the vertices at distance $d(w,u)$ from vertex $u$ that are not closer $v$. Note that different neighbors $v$ of $u$ result in different $|{S_v}^\prime|$. The algorithm evaluates all possible choice of parameter $|{S_v}^\prime|$. }
    
    {$S_v'\leftarrow \{w|d(w,u)=k \wedge d(w,u) \leq d(w,v) \}$}
    
    $X_u^4 \leftarrow \mathrm{sgn}\left(\sum_{w;d(w,u)\leq k}(-1)^{d(w,u)} {|S_v'|}^{-0.5}Y_w\right)$\;
    }
    \KwOut{the vector $X$}
    }
\end{algorithm}
\vspace{4mm}

The classical $k$-local algorithm introduced in \cite{barak2022classical} is designed for high-girth regular graphs. Basso et al.~\cite{basso2022quantum} compared QAOA at $p=11$ with this algorithm in their work for high-girth regular graphs. We also evaluate this classical algorithm in calculating the expected cut fraction on additive product graphs. Due to the different graph structures, there are several natural variations of the original algorithm in \algo{barak}. 

We still assign independent Gaussian random variables to each vertex. When considering the expectation of each edge being cut, the spins of the vertices $u,v$ at each edge $(u,v)$ are determined by a weighted sum of the random variables of the vertices within distance $k$ (if this algorithm is $k$-local) from $u,v$. The vertices at the same distance have the same coefficients in the summation. In \algo{barak}, the coefficient for summing the vertices $w$ at distance $d(w,u)$ from vertex $u$ is $(-1)^{d(w,u)}(D-1)^{-0.5d(w,u)}$. For edge $(u,v)$, denote the set $S_u = \{{w | d(w, u) =k \wedge d(w,u) < d(w,v)}\}$ as the vertices at distance $k$ from $u$ and closer to $u$. The choice of coefficients ensure that the weighted sum of random variables in $S_u$ is normalized in the summation. It is worth noting that for each edge $(u,v)$, the neighborhood structures of $u$ and $v$ are the same, so we only give $u$ as an example here. The same also applies to the coefficients in the additive product graph. For low-girth graphs, we use a similar idea to design the coefficients in the summation of Gaussian distributions. We expound different coefficient choice in \algo{barak_vari}. The algorithm was implemented in Monte Carlo and repeated $10^7$ times for each choice of parameters. The results are shown in \tab{classical-comp}. 

All the codes for classical algorithms are written in Python 3.10 and are publicly available.\footnote{\href{https://github.com/Macondooo/MaxCut-on-low-girth-graphs/tree/main/Classical-local-algorithms-calculation}{https://github.com/Macondooo/MaxCut-on-low-girth-graphs/tree/main/Classical-local-algorithms-calculation}} The random variables are generated by NumPy library. We run our experiments on a 12 vCPU Intel(R) Xeon(R) Platinum 8255C Processor with 40GB memory.

\begin{table}[H]
    \centering
     \caption{The results of the expected cut fraction in classical algorithms for the graphs depicted in \fig{graph-set1} are presented. For the variations of Barak and Marwaha's algorithm outlined in \algo{barak_vari}, the expected cut fraction has been calculated for each possible choice of the coefficients. 
     }
    \label{tab:classical-comp}
    \begin{tabular}{ccccccccc}
        \hline
        \multirow{2}*{Graph} & \multirow{2}*{$k$} & \multirow{2}*{Threshold algorithm}  &  \multicolumn{6}{c}{Variations of Barak and Marwaha's algorithm}  \\
        \cline{4-9}
        ~ & ~ & ~ & Para1 & Para2 & \multicolumn{2}{c}{Para3} & \multicolumn{2}{c}{Para4} \\
        \hline
        \multirow{3}*{\fig{a}} & 1 & 0.601 & 0.639 & \textbf{0.644} & 0.627 & 0.639 &  0.639 & 0.639 \\
        \cline{2-9}
        ~ & 2 & 0.675 & \textbf{0.682} & 0.677 & 0.678 & 0.674 & 0.678 & 0.674 \\
        \cline{2-9}
        ~ & 3 & \textbf{0.713} & 0.701 & 0.694 & 0.692 & 0.695 & 0.695 & 0.695 \\
        \hline
        
        \multirow{3}*{\fig{b}} & 1 & 0.640 & 0.664 & \textbf{0.666} & \multicolumn{2}{c}{0.664} & \multicolumn{2}{c}{0.664} \\
        \cline{2-9}
        ~ & 2 & 0.718 & \textbf{0.733} & 0.725 & \multicolumn{2}{c}{0.727} & \multicolumn{2}{c}{0.727}  \\
        \cline{2-9}
        ~ & 3 & \textbf{0.766} & 0.765 & 0.752 & \multicolumn{2}{c}{0.757} & \multicolumn{2}{c}{0.757} \\
        \hline

        \multirow{3}*{\fig{c}} & 1 & 0.645 & 0.643 & \textbf{0.655} & 0.609 & 0.643 & 0.643 & 0.643 \\
        \cline{2-9}
        ~ & 2 & 0.657 & \textbf{0.685} & 0.676 & 0.668 & 0.659 & 0.676 & 0.659 \\
        \cline{2-9}
        ~ & 3 & \textbf{0.713} & 0.698 & 0.684 & 0.661 & 0.683 & 0.680 & 0.683 \\
        \hline

    \end{tabular}
\end{table}

In \tab{classical-comp}, the results show that when $p$ is smaller, variations of Barak and Marwaha's algorithm achieves better cut fraction. On the other hand, when $p$ is larger, the threshold algorithm can perform better. We then summarize the results of the best-known classical local algorithms in \tab{cut-fraction} to enable a comparison with ma-QAOA and QAOA.

\begin{table}[H]
    \centering
    \caption{The optimal threshold $\tau$ for $p=1,2,3$ in the threshold algorithm (\algo{threshold-algo}).}
    \label{tab:para-threshold}  
    \begin{tabular}{ccc}
        \hline
        Graph & $p$ & $\tau$  \\
        \hline
        \multirow{3}*{\fig{a}} & 1 &  [2 or 3] \\
        \cline{2-3}
        ~ & 2 &  [3,2] \\
        \cline{2-3}
        ~ & 3 &  [2,3,2] \\
        \hline
        
        \multirow{3}*{\fig{b}} & 1 &  [2] \\
        \cline{2-3}
        ~ & 2 &  [2,3]\\
        \cline{2-3}
        ~ & 3 & [2,3,2]\\
        \hline

        \multirow{3}*{\fig{c}} & 1 &  [2]\\
        \cline{2-3}
        ~ & 2 & [2,2] \\
        \cline{2-3}
        ~ & 3 &  [1,2,2]\\
        \hline

    \end{tabular}
\end{table}

\subsection{Numerical evaluation on QAOA}
For ma-QAOA we use our iterative formulas (Eqs.~\eq{G} and \eq{EA}) derived in \sec{theories} to calculate the expectation of cut fraction \eq{cut-fraction} and optimize for $\boldsymbol{\gamma},\boldsymbol{\beta}$. The choice of $\boldsymbol{\gamma}$ vectors for different edges are listed below.
\begin{itemize}
    \item \fig{a} corresponds to $C_3 \oplus C_2 \oplus C_2 \oplus C_2$. Due to symmetry, a $\boldsymbol{\gamma}$ vector is used for edges from $C_3$ and another $\boldsymbol{\gamma}$ vector is used for the remaining $C_2$.
    \item \fig{b} corresponds to $C_4 \oplus C_4$, necessitating two distinct $\boldsymbol{\gamma}$ vectors, one for each component.
    \item \fig{c} corresponds to $C_3 \oplus C_3 \oplus C_2 \oplus C_2 \oplus C_2$. Due to symmetry, one $\boldsymbol{\gamma}$ vector is assigned to the edges from the two $C_3$, and another is assigned to the remaining $C_2$ components.
\end{itemize}

In QAOA, we simply replace the $\Gamma_1,\Gamma_2,\dots,\Gamma_c$ with the same $\Gamma$ in the iterative formula. Our codes are publicly available,\footnote{\href{https://github.com/Macondooo/MaxCut-on-low-girth-graphs/tree/main/QAOA-calculation}{https://github.com/Macondooo/MaxCut-on-low-girth-graphs/tree/main/QAOA-calculation}} which are written in Python 3.10. Our experiments for ma-QAOA and QAOA are obtained by simulations on classical computers, with the computational environment same to that used for the classical algorithms.

To accelerate the computations, we leverage vectorization techniques to evaluate $\eq{G}$ and $\eq{EA}$. In this approach, each term in the formulas is treated as a tensor. The exponential, for instance, is regarded as a tensor with subscript ${\boldsymbol{z}_k}$.
In addition, parameters $\boldsymbol{\gamma},\boldsymbol{\beta}$ are optimized by the L-BFGS-B algorithm in SciPy library. Memory is dominated by store both the iterative formula \eq{G} and the precomputed tensor. 
The results of ma-QAOA and QAOA and its comparison with the best-known classical local algorithms in \sec{comp-classical} are listed in \tab{cut-fraction}.

\begin{table}[!htbp]
    \centering
     \caption{Experiment results for ma-QAOA, QAOA, and best-known classical local algorithms. All the algorithms calculated the expected cut fraction. The column ``Best-known classical local algorithms'' shows the highest expected cut fraction of all best-known classical local algorithms for each graph and each $k$ in \tab{classical-comp}. }
    \begin{tabular}{ccccc}
        \hline
        Graph & $p,k$ & ma-QAOA & QAOA  & Best-known classical local algorithms  \\
        \hline
        \multirow{3}*{\fig{a}} & 1 & \textbf{0.65172} & 0.64590  & 0.644  \\
        \cline{2-5}
        ~ & 2 & \textbf{0.69754} & 0.69335 & 0.682 \\
        \cline{2-5}
        ~ & 3 & \textbf{0.72506} & 0.72071 &  0.713 \\
        \hline
        
        \multirow{3}*{\fig{b}} & 1 & 0.66238 & 0.66238 & \textbf{0.666} \\
        \cline{2-5}
        ~ & 2 & \textbf{0.74999} & 0.72780 & 0.733 \\
        \cline{2-5}
        ~ & 3 & 0.76414 & 0.76414 &  \textbf{0.766} \\
        \hline

        \multirow{3}*{\fig{c}} & 1 & \textbf{0.67642} & 0.65991 &  0.655 \\
        \cline{2-5}
        ~ & 2 & \textbf{0.72485} & 0.72084  &  0.685 \\
        \cline{2-5}
        ~ & 3 & \textbf{0.75435} & 0.74554  &  0.713 \\
        \hline

    \end{tabular}
    \label{tab:cut-fraction}
\end{table}

\tab{cut-fraction} presented the expected cut fraction of the graphs shown in \fig{graph-set1}, as obtained through the ma-QAOA, QAOA, and classical local algorithms in \sec{comp-classical}. Due to the limitations in computational capacity, we tested both algorithms under values of $p,k$ up to 3. 
The results show that on the additive product graphs in \fig{a} and \fig{c}, QAOA performs a 0.3-5.2 percentage advantage over classical local algorithms. In the meantime, ma-QAOA can improve this advantage by 0.6-2.5 percentage.
More importantly, while the classical local algorithms exhibit a minor advantage over QAOA on the graphs shown in \fig{b}, the ma-QAOA demonstrates the advantages by more precisely encoding the underlying graph structure.
As mentioned previously, \fig{b} is constructed from square-based atomic structures. We can categorize the edges into two distinct sets and assign different angle parameters to these two edge sets. Intuitively, if a given square is assigned angle parameters of $(\boldsymbol{\gamma}_{1,1},\boldsymbol{\gamma}_{2,1},\dots,\boldsymbol{\gamma}_{p,1})$, then the four adjacent squares will have angle parameters of $(\boldsymbol{\gamma}_{1,2},\boldsymbol{\gamma}_{2,2},\dots,\boldsymbol{\gamma}_{p,2})$. The converse also holds true. The squares with $(\boldsymbol{\gamma}_{1,2},\boldsymbol{\gamma}_{2,2},\dots,\boldsymbol{\gamma}_{p,2})$ are surrounded by squares with $(\boldsymbol{\gamma}_{1,1},\boldsymbol{\gamma}_{2,1},\dots,\boldsymbol{\gamma}_{p,1})$.
When $p=1$, the ma-QAOA is unable to fully discern the complete square structure in the graph, and thus provides no performance improvement. However, for $p=2$ and $3$ in \fig{b}, the ma-QAOA can recognize the square structures in the graph. Consequently, the use of two separate angle parameter vectors, $\boldsymbol{\gamma}_1$ and $\boldsymbol{\gamma}_2$, enables the ma-QAOA to better capture the graph structure, leading to better performance compared to both QAOA and the best-known classical local algorithms. 
When $p=3$ in \fig{b}, the performance of ma-QAOA and QAOA is identical. This may be attributed to the complexity of the optimization landscape, which causes random initial values failing to converge to the global optimum.

The optimal values $\boldsymbol{\gamma},\boldsymbol{\beta}$ of ma-QAOA and QAOA are shown in \tab{para-ma-qaoa} and \tab{para-qaoa} respectively. We find that the optimized parameters at lower circuit depths can serve as good initial values for optimization at higher depths. Specifically, the optimized $\boldsymbol{\gamma}$ and $\boldsymbol{\beta}$ at $p=1$ are close to the optimized ones at $p=2$, and the same holds for the parameters at $p=2$ and $p=3$.

\begin{table}[H]
    \centering
    \caption{The optimal $\boldsymbol{\gamma},\boldsymbol{\beta}$ for $p=1,2,3$ in ma-QAOA}
    \label{tab:para-ma-qaoa}  
    \begin{tabular}{cccc}
        \hline
        Graph & $p$ & $\boldsymbol{\gamma}$ & $\boldsymbol{\beta}$  \\
        \hline
        \multirow{3}*{\fig{a}} & 1 & [0.1901, 0.3195] & [0.3814] \\
        \cline{2-4}
        ~ & 2 & [0.1440, 0.2494, 0.2880, 0.4576] & [0.5221, 0.2725] \\
        \cline{2-4}
        ~ & 3 & [0.1156, 0.2294, 0.2565, 0.3943, 0.3047, 0.4685] & [0.5753, 0.4004, 0.2139] \\
        \hline
        
        \multirow{3}*{\fig{b}} & 1 & [0.2617, 0.2617] & [0.3927] \\
        \cline{2-4}
        ~ & 2 & [0, 0.4522, 0, 0.5592] & [0.5592, 0.4522] \\
        \cline{2-4}
        ~ & 3 & [0.1698, 0.1698, 0.3047, 0.3047, 0.3576, 0.3576] & [0.6037, 0.4643, 0.2422]  \\
        \hline

        \multirow{3}*{\fig{c}} & 1 & [0.1908, 0.5021] & [0.3835] \\
        \cline{2-4}
        ~ & 2 & [0.2390, 0.2763, 0.4102, 0.6389] & [0.4663, 0.2589] \\
        \cline{2-4}
        ~ & 3 & [0, 0.7853, 1.3256, 1.3615, 1.5708, 0.3415] & [0.3926, 0.7853, 0.7854] \\
        \hline

    \end{tabular}
\end{table}

\begin{table}[H]
    \centering
    \caption{The optimal $\boldsymbol{\gamma},\boldsymbol{\beta}$ for $p=1,2,3$ in QAOA}
    \label{tab:para-qaoa}  
    \begin{tabular}{cccc}
        \hline
        Graph & $p$ & $\boldsymbol{\gamma}$ & $\boldsymbol{\beta}$  \\
        \hline
        \multirow{3}*{\fig{a}} & 1 & [0.2536] & [0.3662] \\
        \cline{2-4}
        ~ & 2 & [0.2048, 0.3902] & [0.4857, 0.2578] \\
        \cline{2-4}
        ~ & 3 & [0.1786, 0.3448, 0.4000] & [0.5506, 0.3738, 0.2016] \\
        \hline
        
        \multirow{3}*{\fig{b}} & 1 & [0.2617] & [0.3927] \\
        \cline{2-4}
        ~ & 2 & [0.1978, 0.3534] & [0.5557, 0.3133] \\
        \cline{2-4}
        ~ & 3 & [0.1698, 0.3047, 0.3576] & [0.6037, 0.4643, 0.2422] \\
        \hline

        \multirow{3}*{\fig{c}} & 1 & [0.2851] & [0.3481] \\
        \cline{2-4}
        ~ & 2 & [0.2917 0.5623] & [0.4090 0.2408] \\
        \cline{2-4}
        ~ & 3 & [0.2014 0.48550 0.5916] & [0.5410 0.3187 0.1851] \\
        \hline
    \end{tabular}
\end{table}

\subsection{Additional experiments for low-girth graphs}\label{sec:other graph}

In this subsection, we analyze a set of graphs that differ from the Ramanujan and high-girth regular graphs previously discussed. Specifically, here we consider graphs with more mesh-like structures, rather than tree-like structures in  Ramanujan graphs and high-girth regular graphs. This distinction in the underlying topology of these graph families may produce interesting results when analyzing and comparing the performance of the QAOA and classical optimization techniques. 
The graphs tested are depicted in \fig{tiling-graph} and the results are listed in \tab{tiling-comp}. The running environments for quantum algorithms and classical local algorithms are identical to the ones in \sec{experiments}.

\begin{figure}[!htbp]
\centering
    \subcaptionbox{\label{fig:5-6cycle}}{\includegraphics[width = 0.3\textwidth,page={1}]{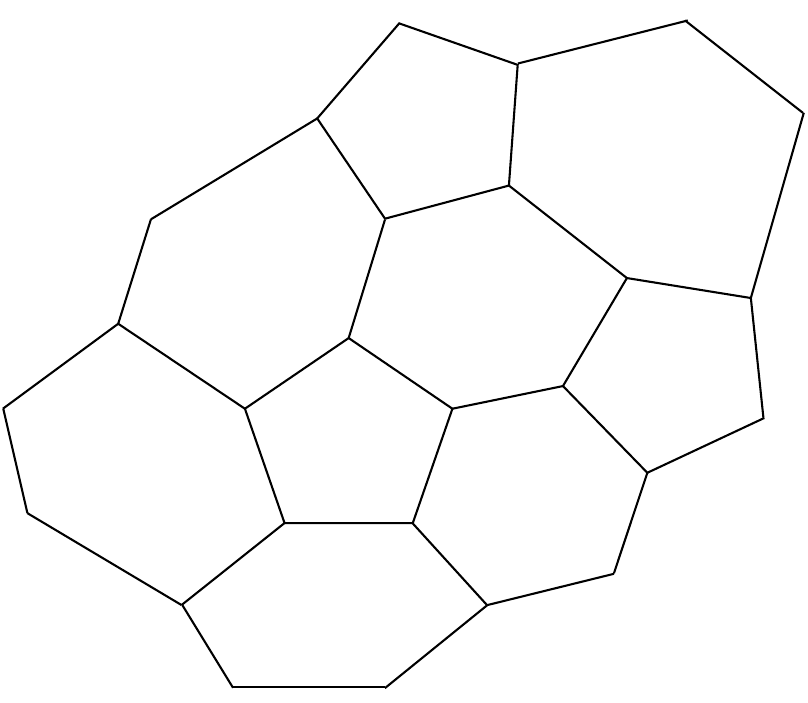}}
    \hspace{6mm}
    \subcaptionbox{\label{fig:3-4-6cycle}}{\includegraphics[width = 0.3\textwidth,page={2}]{mesh_graph.pdf}}
\caption{The tiling grid graphs are evaluated using ma-QAOA, QAOA, and classical local algorithms. \fig{5-6cycle} consists of pentagons and hexagons. Each pentagon is connected to five hexagons, and each hexagon is surrounded by three pentagons and three hexagons. In ma-QAOA, two $\boldsymbol{\gamma}$ vectors are employed: one for edges within $C_5$ and $C_6$, and another for edges shared by two $C_6$. \fig{3-4-6cycle} comprises triangles, quadrilaterals, and hexagons. Its vertices correspond to integer vertices from the hexagonal tiling grid. Similarly, there are two types of edges: those in $C_6, C_4$, and those in $C_3, C_4$, with two $\boldsymbol{\gamma}$ vectors applied to these edges.}
\label{fig:tiling-graph}
\end{figure}

\begin{table}[H]
    \centering
     \caption{The expected cut fraction of the graphs in \fig{tiling-graph} are presented. The results of ma-QAOA, QAOA, and best-known classical local algorithms are compared. }
    \label{tab:tiling-comp}
    \resizebox{\linewidth}{!}{
    \begin{tabular}{ccccccccccc}
        \hline
        \multirow{2}*{Graph} & \multirow{2}*{$k,p$} &  \multirow{2}*{ma-QAOA} &\multirow{2}*{QAOA} & \multirow{2}*{Threshold algorithm}  &  \multicolumn{6}{c}{Variations of Barak and Marwaha's algorithm}  \\
        \cline{6-11}
        ~ & ~ & ~ & ~ & ~ & Para1 & Para2 & \multicolumn{2}{c}{Para3} & \multicolumn{2}{c}{Para4} \\
        \hline
        
        \multirow{2}*{\fig{5-6cycle}} & 1 & 0.69245 & 0.69245 & 0.687 & 0.691 & \textbf{0.695} & 0.691 & 0.691 & 0.691 & 0.691 \\
        \cline{2-11}
        ~ & 2 & \textbf{0.75296} & 0.75243 & 0.730 &  0.736 & 0.734 & 0.721 & 0.736 & 0.736 & 0.736 \\
        \hline
        
        \multirow{2}*{\fig{3-4-6cycle}} & 1 & \textbf{0.65172} & 0.64589 & 0.601 & 0.639 & 0.644 & 0.639 & 0.627 & 0.639 & 0.639 \\
        \cline{2-11}
        ~ & 2 & \textbf{0.70420} & 0.69970 & 0.686 &  0.691 & 0.679 & 0.672 & 0.676 & 0.672 & 0.672 \\
        \hline
    \end{tabular}
    }
\end{table}

\tab{tiling-comp} shows that the performance of QAOA is better than the best-known classical local algorithms for the tiling grid graphs by 0.3-2.2 percentage except $p=1$ in \fig{5-6cycle}. The ma-QAOA further enhances this advantage by 0.1-0.9 percentage.

\section{Conclusions}
In this paper, we systematically investigated the performance guarantee of QAOA for MaxCut on low-girth graphs, and we use ma-QAOA to more precisely encoding the structure of low girth graph. In theory, 
we calculated the expected cut fraction of QAOA for the MaxCut on a set of expander graphs known as additive product graphs~\cite{mohanty2020x}, which was achieved by analyzing the structure of subgraphs in additive product graphs iteratively. Furthermore, we also extend our framework to the quantum MaxCut problem on additive product graphs.
In experiments, to construct convincing benchmarks against QAOA, we also investigated the best-known classical local algorithms for the MaxCut problem and tested them on low-girth additive product graphs. Specifically, we analyzed the threshold algorithm presented in \algo{threshold-algo} as well as the variations of the algorithm proposed by Barak and Marwaha \cite{barak2022classical} in \algo{barak_vari}. 
The results indicate that QAOA outperforms the best-known classical algorithms on additive product graphs, and ma-QAOA further enhances this advantage.

Furthermore, the performance comparison between QAOA and classical local algorithms are extended to the tiling grid graphs. On these types of graphs, the ma-QAOA and QAOA still demonstrates advantages over classical algorithms.

Our work leaves open questions that necessitate further investigation:
\begin{itemize}
    \item The performance of QAOA on other low-girth expander graphs is not yet fully clear. In particular, this paper leaves room for further exploration into how properties such as vertex expansion, edge expansion, or spectral expansion may affect the performance gap between QAOA and classical local algorithms.
    \item Still, the lower bound of QAOA's performance on general graphs is unknown. Previous studies mostly focused on regular graphs \cite{basso2022quantum,wurtz2021maxcut} or complete graphs \cite{boulebnane2021predicting} with different edge weights distribution. Specifically, understanding how QAOA's theoretical bounds compared to the 0.878 classical approximation guarantee for general graphs by the Goemans-Williamson algorithm~\cite{goemans1995improved} is an important open question. Exploring this comparison could shed light on the relative strengths of quantum and classical approaches to approximate optimization problems. 
\end{itemize}

\section*{Acknowledgments}
TL, YS, and ZY were supported by the National Natural Science Foundation of China (Grant Numbers 92365117 and 62372006), and The Fundamental Research Funds for the Central Universities, Peking University.

\newcommand{\arxiv}[1]{arXiv:\href{https://arxiv.org/abs/#1}{\ttfamily{#1}}\?}\newcommand{\arXiv}[1]{arXiv:\href{https://arxiv.org/abs/#1}{\ttfamily{#1}}\?}\def\?#1{\if.#1{}\else#1\fi}
\providecommand{\bysame}{\leavevmode\hbox to3em{\hrulefill}\thinspace}

\end{document}